\documentclass[11pt, twoside, letterpaper]{amsart}
\usepackage{amsmath, hyperref, color, amssymb}
\usepackage{srcltx}

\newtheorem{thm}{Theorem}[section]
\newtheorem{cor}[thm]{Corollary}
\newtheorem{lem}[thm]{Lemma}
\newtheorem{prop}[thm]{Proposition}

\theoremstyle{definition}
\newtheorem{defn}[thm]{Definition}

\newtheorem{ass}[thm]{Assumption}
\newtheorem{stass}{Standing Assumption}
\theoremstyle{remark}
\newtheorem{rem}[thm]{Remark}
\newtheorem{exa}[thm]{Example}
\numberwithin{equation}{section}

\newcommand{\R}{\mathbb{R}}
\newcommand{\N}{\mathbb{N}}

\newcommand{\A}{\mathcal{A}}

\newcommand{\prob}{\mathbb{P}}
\newcommand{\M}{\mathcal{M}}

\newcommand{\cE}{\mathcal{B}_E}
\newcommand{\FF}{\bF}

\newcommand{\qprob}{\mathbb{Q}}
\newcommand{\expec}{\mathbb{E}}
\newcommand{\basisp}{(\Omega,\bF,\prob)}
\newcommand{\basisq}{(\Omega,\bF,\qprob)}
\newcommand{\basisgp}{(\Omega,\bG,\prob)}
\newcommand{\F}{\mathcal{F}}
\newcommand{\G}{\mathcal{G}}

\newcommand{\bH}{\mathbf{H}}
\newcommand{\ud}{\,\mathrm d}

\newcommand{\cO}{\mathcal{O}}
\newcommand{\cP}{\mathcal{P}}
\newcommand{\Mloc}{\M_{\text{loc}}}
\newcommand{\mMloc}{\M_{\text{\emph{loc}}}}

\DeclareMathOperator{\esssup}{ess\,sup}
\DeclareMathOperator{\argmax}{arg\,max}

\newcommand{\pare}[1]{\left(#1\right)}

\newcommand{\dbra}[1]{[\kern-0.15em[ #1 ]\kern-0.15em]}
\newcommand{\dbraco}[1]{[\kern-0.15em[ #1 [\kern-0.15em[}
\newcommand{\dbraoc}[1]{]\kern-0.15em] #1 ]\kern-0.15em]}
\newcommand{\dbraoo}[1]{]\kern-0.15em] #1 [\kern-0.15em[}

\newcommand{\Z}{\mathcal{Z}}

\newcommand{\bF}{\mathbf{F}}
\newcommand{\Cons}{\mathcal{C}}

\newcommand{\bG}{\mathbf{G}}
\newcommand{\GG}{\mathbf{G}}
\newcommand{\ind}{\mathbf{1}}

\newcommand{\be}{\begin{equation}}
\newcommand{\ee}{\end{equation}}
\newcommand{\ba}{\begin{aligned}}
\newcommand{\ea}{\end{aligned}}

\begin{document}

\title{The value of informational arbitrage}

\author{Huy N. Chau}
\address{Huy N. Chau, Alfr\'ed R\'enyi Institute of Mathematics, Hungarian Academy of Sciences (Hungary)}
\email{chau@renyi.hu}

\author{Andrea Cosso}
\address{Andrea Cosso, Dipartimento di Matematica, Universit\`a di Bologna (Italy)}
\email{andrea.cosso@unibo.it}

\author{Claudio Fontana}
\address{Claudio Fontana, 
LPSM, Paris Diderot University (France)}
\email{fontana@math.univ-paris-diderot.fr}

\thanks{The authors are thankful to Albina Danilova, 
Martin Larsson, Martin Schweizer and seminar participants at ETH Z\"urich and LSE for valuable discussions and suggestions on the topic of the present paper.
Huy N. Chau was supported by the ``Lend\"{u}let'' grant LP2015-6 of the Hungarian Academy of Sciences and by the NKFIH (National Research, Development and Innovation Office, Hungary) 
	grant KH 126505.}

\date{\today}


\maketitle

\begin{abstract}
In the context of a general semimartingale model of a complete market, we aim at answering the following question: How much is an investor willing to pay for learning some inside information that allows to achieve arbitrage? 
If such a value exists, we call it the value of informational arbitrage.
In particular, we are interested in the case where the inside information yields arbitrage opportunities but not unbounded profits with bounded risk. 
In the spirit of Amendinger et al. (2003, {\em Finance Stoch.}), we provide a general answer to the above question by relying on an indifference valuation approach. 
To this effect, we establish some new results on models with inside information and study optimal investment-consumption problems in the presence of initial information and arbitrage, also allowing for the possibility of leveraged positions.
We characterize when the value of informational arbitrage is universal, in the sense that it does not depend on the preference structure. 
Our results are illustrated with several explicit examples.
\end{abstract}

\section{Introduction}		\label{sec:intro}

The notion of information plays a crucial role in the analysis of investment decisions. In line with economic intuition, access to more precise sources of information gives an informational advantage leading to better performing portfolios.
The problem of quantifying such an informational advantage represents a central question in finance and has constantly attracted significant attention in financial economics and, more recently, in mathematical finance.

In this work, we develop a general approach for quantifying in monetary terms the informational advantage associated to some inside information, in the context of a general semimartingale model of a complete market, under weak assumptions on the random variable (denoted by $L$) representing the inside information.
We adopt an indifference valuation approach, determining a value $\pi(v)$ which makes a risk averse agent with initial capital $v$ indifferent between the following two alternatives: (i) invest optimally the initial capital $v$ by relying on the publicly available information only; (ii) acquire the inside information $L$ at the price $\pi(v)$ and invest optimally the residual capital $v-\pi(v)$ by relying on the publicly available information enriched by the inside information.

The idea of quantifying information through an indifference valuation approach can be traced back to early contributions in information economics, see in particular \cite{LaValle,MR72,Morris,Willinger}. 
The same approach has been pursued in the context of modern mathematical finance in \cite{ABS}, which represents the main starting point for the present work.
An alternative approach for the measurement of the value of information consists in computing the {\em utility gain} of an informed agent, as considered for instance in \cite{PK96,AIS98,Hillairet,MR2223957,HJ}. However, this has the drawback of expressing the value of information in utility terms, and not in monetary units.

A distinguishing feature of the present work is that we explicitly allow the inside information $L$ to generate arbitrage opportunities for an informed agent. 
This situation is not covered by the existing literature, except for the extreme case where the knowledge of $L$ is so informative that it leads to infinite utility for an informed agent (see, e.g., \cite{PK96,AIS98}).
In contrast, we assume that the inside information can be exploited  to realize arbitrage opportunities, but unbounded profits with bounded risk cannot be achieved (this represents the minimal condition allowing for a meaningful solution to optimal portfolio problems, see \cite{MR2335830,CDM15,CCFM2015}).
In this framework, we call the indifference value of $L$ the {\em value of informational arbitrage}.

As we are going to show, informational arbitrage arises whenever the inside information reveals that some events, which are believed to occur with strictly positive probability by public opinion, are actually impossible.
In order to illustrate the notion of the value of informational arbitrage, let us present a simple example, which will be analysed in a more general version in Section \ref{sec:example_Brownian}. 

\begin{exa}	\label{ex:intro}
Consider a financial market with a single risky asset, with price process 
\[
S_t = \exp\left(W_t-t/2\right),
\qquad\text{ for all }t\in[0,1],
\]
where $(W_t)_{t\in[0,1]}$ is a standard Brownian motion.
The {\em ordinary information} (publicly available) is given by the observation of the price process alone, corresponding to the filtration $\bF=(\F_t)_{t\in[0,1]}$.
We suppose that the {\em inside information} is represented by the observation at $t=0$ of the random variable $L=\ind_{\{W_1\geq0\}}$. In other words, an informed agent who observes the realization of $L$ knows before the beginning of trading  whether the terminal value $S_1$ of the asset will be above or below the threshold $1/\sqrt{e}$. 
The information flow available to an informed agent is described by the  initially enlarged filtration $\bG=(\G_t)_{t\in[0,1]}$, where $\G_t=\F_t\vee\sigma(L)$ for all $t\in[0,1]$.

Clearly, the ordinary information does not allow any kind of arbitrage profits and every risk averse agent would simply choose not to trade in the risky asset $S$. On the contrary, the inside information $L$ yields arbitrage opportunities, which can also be realized through suitably chosen buy-and-hold strategies. 
In this sense, we say that $L$ yields {\em informational arbitrage} and we aim at determining the indifference value $\pi(v)$, namely the maximal amount that an agent with initial wealth $v>0$ accepts to pay for learning the realization of $L$ before the beginning of trading.

In the context of this example, we will show that
for any risk averse agent constrained to invest in non-negative portfolios the value of informational arbitrage is always given by
\[
\pi(v) = v/2.
\]
This means that, while it is attractive to acquire the inside information $L$, a risk averse agent would not sacrifice more than one half of his initial wealth in order to have the possibility of achieving arbitrage.
Moreover, there exists an arbitrage strategy which represents the optimal trading strategy for every risk averse informed agent.
Referring to Section \ref{sec:example_Brownian} for a detailed analysis of this example, we point out that the value $\pi(v)$ presents the striking feature of being a universal indifference value, which does not depend on the preference structure.
\end{exa}

In the present work, we aim at revealing which features of the inside information are at the origin of the appearance of arbitrage and understanding the indifference value of informational arbitrage in a general setting.
Motivated by Example \ref{ex:intro} and similarly as in \cite{ABS}, the problem is naturally framed in the context of an initial enlargement of filtration.
However, in order to allow for the possibility of informational arbitrage, we have to depart from the conventional assumption that $L$ is independent of the ordinary information flow $\bF$ under an equivalent probability measure (called {\em decoupling measure} in \cite{ABS}).
The notion of decoupling measure goes back to early works in the theory of enlargement of filtrations and has been widely employed in the insider trading literature (see, e.g., \cite{GP98,MR1831271,Amen,Campi05,Hillairet,HJ}). 
The existence of a decoupling measure is tantamount to the equivalence between the $\bF$-conditional law of $L$ and its unconditional law and allows to easily transfer to the initially enlarged filtration $\bG$ most of the properties of $\bF$, including the (semi-)martingale property, market completeness and, most notably, absence of arbitrage.

We assume the validity of  {\em Jacod's density hypothesis}, as introduced in the seminal paper \cite{Jac85}. This condition is significantly weaker than the existence of a decoupling measure, as it corresponds to the absolute continuity (but not necessarily equivalence) of the $\bF$-conditional law of $L$ with respect to its unconditional law.
While the passage from an equivalence to an absolute continuity relation could appear as a technical generalization, it turns out to require the development of a novel approach. Most importantly, it allows the inside information to generate arbitrage, as shown in Example \ref{ex:intro}, thus covering situations that cannot be addressed by the existing works.

The main results and contributions of the paper can be outlined as follows.
First, we provide a new martingale representation result in the initially enlarged filtration $\bG$, showing that market completeness can be transferred from $\bF$ to $\bG$ up to a change of num\'eraire. By relying on this result, we obtain a complete characterization of the validity of {\em no free lunch with vanishing risk} (NFLVR, see \cite{DS94,DS98b}) and {\em no unbounded profit with bounded risk} (NUPBR, see \cite{MR2335830}) in $\bG$.
This set of theoretical results provides the necessary foundations for the study of optimal consumption-investment problems under inside information and, possibly, in the presence of arbitrage opportunities. 
We consider preferences described by a utility stochastic field together with a random consumption clock and assume that agents are allowed to enter into leveraged positions, up to the limit of a fixed credit line. 
We provide a general solution via duality methods, which reveals the interplay between arbitrage and leverage and can be made explicit in the case of typical utility functions.
In turn, this enables us to study the indifference value $\pi(v)$ of the inside information $L$. Under natural assumptions, we prove that $\pi(v)$ is finite and also strictly positive and increasing in the allowable leverage whenever the information revealed by $L$ allows to achieve arbitrage opportunities, regardless of the preference structure.
For logarithmic and power utility functions, we obtain explicit expressions for $\pi(v)$, thereby generalizing the results of \cite{ABS} and proving some of the empirical findings reported in \cite{Liu_et_al}. 
Moreover, we provide universal bounds for the value of informational arbitrage and characterize when the indifference value of inside information is a {\em universal} value which does not depend on the preference structure, as in the case of Example \ref{ex:intro}. In particular, we show that this can happen only in the presence of arbitrage.


\subsection{Structure of the paper}

In Section \ref{sec:setting}, we introduce the general setting, consisting of two financial markets associated to two different filtrations. We provide a new martingale representation result and study (no-)arbitrage properties in the presence of inside information. 
Section \ref{sec:portfolio} deals with optimal consumption-investment problems under general preferences, allowing for non-trivial initial information, leverage and arbitrage. 
In Section \ref{sec:indiff_price} we study the indifference value of inside information and characterize under which conditions it is a universal value. Section \ref{sec:examples} contains three examples. 
For better readability, the proofs are collected in Section \ref{sec:proofs}.
Section \ref{sec:conclusions} concludes by discussing the role of  market completeness and pointing out possible directions of further research.

\subsection{Notation}

Throughout the paper, we adopt the following conventions and notations, referring to \cite{MR1219534,MR1943877} for all unexplained notions related to stochastic calculus. Let $(\Omega,\A,\prob)$ be a generic probability space endowed with some filtration $\mathbf{H}=(\mathcal{H}_t)_{t\in[0,T]}$ satisfying the usual conditions of right-continuity and $\prob$-completeness, with $T\in(0,+\infty)$ a fixed time horizon. 
We denote by $\M(\prob,\mathbf{H})$ ($\Mloc(\prob,\mathbf{H})$, resp.) the set of martingales (local martingales, resp.) on $(\Omega,\mathbf{H},\prob)$ and we tacitly assume that every local martingale has c\`adl\`ag paths a.s. 
For a given $\R^d$-valued semimartingale $X=(X_t)_{t\in[0,T]}$ on $(\Omega,\mathbf{H},\prob)$, we denote by $L(X,\mathbf{H})$ the set of all $\mathbf{H}$-predictable $\R^d$-valued processes $\varphi=(\varphi_t)_{t\in[0,T]}$ which are integrable with respect to $X$ in the filtration $\mathbf{H}$. Recall that the set $L(X,\mathbf{H})$ is invariant under equivalent changes of measure (see, e.g., \cite[Theorem 12.22]{MR1219534}). The stochastic integral of $\varphi\in L(X,\mathbf{H})$ with respect to $X$ is denoted by $(\varphi\cdot X)_t:=\int_{(0,t]}\varphi_u\ud X_u$, for all $t\in[0,T]$, with $(\varphi\cdot X)_0=0$.
Finally, we denote by $\cO(\mathbf{H})$ and $\cP(\mathbf{H})$, respectively, the optional and predictable sigma-fields on $\Omega\times[0,T]$ with respect to the filtration $\mathbf{H}$.
For an adapted process $Y=(Y_t)_{t\in[0,T]}$, we write $Y\in\cO_+(\mathbf{H})$ to denote that $Y$ is a non-negative $\cO(\mathbf{H})$-measurable process.

\section{The ordinary and the insider financial markets}	\label{sec:setting}

In this section, we first present the ordinary financial market (Section \ref{sec:ordinary}), consisting of a general arbitrage-free complete financial market in a reference filtration $\bF$. In Section \ref{sec:G}, we introduce the initially enlarged filtration $\bG$ associated to the inside information and state a new martingale representation result in $\bG$. In Section \ref{sec:arbitrage_in_G}, we characterize the (no-)arbitrage properties of the financial market under inside information.

\subsection{The ordinary financial market}	\label{sec:ordinary}

We consider a probability space $(\Omega,\A,\prob)$ endowed with a filtration $\bF = (\F_t)_{t\in[0,T]}$ satisfying the usual conditions, where $T<+\infty$ represents a fixed investment horizon. 
For simplicity of presentation, we assume that the initial sigma-field $\F_0$ is trivial.
On $(\Omega,\bF,\prob)$, we let $S=(S_t)_{t\in[0,T]}$ be a $d$-dimensional non-negative semimartingale, representing the prices of $d$ risky assets, discounted with respect to some baseline security\footnote{The non-negativity assumption on $S$ is made for simplicity of presentation and can be relaxed at the expense of slightly greater technicalities, introducing the notion of sigma-martingale (see, e.g., \cite{DS98b,TakSch}).}.

We call {\em ordinary financial market} the tuple $(\Omega,\bF,\prob;S)$, where the filtration $\bF$ is supposed to represent the publicly available information.
We assume that $S$ satisfies the {\em no free lunch with vanishing risk} (NFLVR) condition on $\basisp$, see \cite{DS98b}. More specifically, we shall assume the validity of the following condition throughout the paper.

\begin{stass}	\label{ass:Q}
There exists a unique probability measure $\qprob$ on $(\Omega,\F_T)$ such that $\qprob\sim\prob$ and $S\in\Mloc(\qprob,\bF)$.
\end{stass}

Assumption \ref{ass:Q} implies that the ordinary financial market $(\Omega,\bF,\prob;S)$ is arbitrage-free (in the sense of NFLVR) and {\em complete}.\footnote{The relevance of the market completeness assumption will be discussed in Section \ref{sec:conclusions}.} 
Indeed, by \cite[Theorem 11.3]{MR542115}, every local martingale on $(\Omega,\bF,\qprob)$ can be represented as a stochastic integral of $S$. In particular, every $\F_T$-measurable (bounded) contingent claim can be replicated by self-financing trading.
We denote by $Z=(Z_t)_{t\in[0,T]}$ the density process of $\qprob$ with respect to $\prob$ on $\bF$, i.e., $Z_t=\!\ud\qprob|_{\F_t}/\!\ud\prob|_{\F_t}$, for all $t\in[0,T]$. Recall that $Z$ is a strictly positive martingale on $\basisp$ with $Z_0=1$.

\begin{rem}	\label{rem:NUPBR_in_F}
Assumption \ref{ass:Q} can be relaxed by only requiring the existence of a unique equivalent local martingale deflator for $S$ on $\basisp$ (see \cite{MR2972237}). This ensures that NUPBR holds on $(\Omega,\bF,\prob;S)$, so that portfolio optimization problems can be meaningfully solved in $\bF$.
In view of \cite[Corollary 2.1]{SY98}, this condition also suffices to ensure that the ordinary financial market $(\Omega,\bF,\prob;S)$ is complete.
However, since the main goal of the present paper is to study the value of inside information generating arbitrage opportunities, when the latter are impossible to achieve on the basis of ordinary information alone, we find it more natural to work under Assumption \ref{ass:Q}.
\end{rem}

\subsection{The initially enlarged filtration $\bG$}	\label{sec:G}

We assume that the inside information is generated by an $\A$-measurable random variable $L$ taking values in a Lusin space $(E,\cE)$, where $\cE$ denotes the Borel sigma-field of $E$. The associated {\em initially enlarged filtration} $\bG=(\G_t)_{t\in[0,T]}$ is defined as the smallest filtration containing $\bF$ and such that $L$ is $\G_0$-measurable, i.e., $\G_t:=\F_t\vee\sigma(L)$, for all $t\in[0,T]$.
We denote by $\lambda : \cE \rightarrow [0,1]$ the unconditional law of $L$, so that $\lambda(B)=\prob(L\in B)$ holds for all $B\in\cE$. For each $t\in[0,T]$, let $\nu_t : \Omega \times \cE \rightarrow [0,1]$ be a regular version of the $\F_t$-conditional law of $L$ (which always exists since $(E,\cE)$ is Lusin). 

Throughout the paper, we shall assume the validity of the following condition, which is known as {\em Jacod's density hypothesis} in enlargement of filtrations theory (see \cite{Jac85}). 

\begin{stass}	\label{ass:Jacod}
For all $t\in[0,T]$, $\nu_t \ll\lambda$ holds in the a.s. sense.
\end{stass}

Assumption \ref{ass:Jacod} was introduced in the seminal work \cite{Jac85} in order to prove the validity of the {\em $H'$-hypothesis}, namely that every $\bF$-semimartingale is also a $\bG$-semimartingale. In a frictionless financial market, the failure of the semimartingale property is incompatible with NUPBR (see \cite{KarPla11}), which is in turn a necessary condition for the solution of portfolio optimization problems (see \cite[Proposition 4.19]{MR2335830}). Therefore, the validity of the $H'$-hypothesis represents a necessary requirement in our framework (in this respect, see also Remark \ref{rem:density_clock}).

A central feature of our work is that Assumption \ref{ass:Jacod} is only required to hold as an absolute continuity relation and not as an equivalence, as explained in the introduction. This fact turns out to be intimately linked to the existence of arbitrage opportunities in $\bG$ (see Theorem \ref{thm:results_NUPBR}).
As a preliminary, the following lemma presents some fundamental consequences of Assumption \ref{ass:Jacod}.

\begin{lem}	\label{lem:prelim_G}
The filtration $\GG$ is right-continuous and every semimartingale on $\basisp$ is also a semimartingale on $\basisgp$. 
There exists a $\pare{\cE \otimes \cO(\FF)}$-measurable function $E \times \Omega \times [0,T] \ni (x, \omega,t)\mapsto q^x_t(\omega) \in \R_+$, c\`adl\`ag \ in $t \in[0,T]$ and such that:
\begin{itemize}
\item[(i)]
for every $t\in[0,T]$, $\nu_t(\!\ud x) = q^x_t \,\lambda(\!\ud x)$ holds a.s.;
\item[(ii)]
for every $x\in E$, the process $q^x=(q^x_t)_{t\in[0,T]}$ is a martingale on $\basisp$.
\end{itemize}
Furthermore, it holds that $\prob(q^L_t>0)=1$, for all $t\in[0,T]$.
\end{lem}

Under the present standing assumptions, we can establish the following proposition, which provides a fundamental martingale representation result in the initially enlarged filtration $\bG$. 

\begin{prop}	\label{prop:MRT}
Let $M=(M_t)_{t\in[0,T]}$ be a local martingale on $\basisgp$. Then there exists a process $K=(K_t)_{t\in[0,T]}\in L(S,\bG)$ such that
\[
M_t = \frac{Z_t}{q^L_t}\bigl(M_0 + (K\cdot S)_t\bigr)
\qquad\text{ a.s. for all $t\in[0,T]$.}
\]
\end{prop}

The above proposition shows that the martingale representation property of $S$ on $\basisq$ can be transferred to the initially enlarged filtration $\bG$ under $\prob$ up to a suitable ``change of num\'eraire'', represented by the process $Z/q^L$. 
Furthermore, the process $Z/q^L$ plays a key role in the study of the (no-)arbitrage properties of $S$ on $\basisgp$, as we are going to see in Section \ref{sec:arbitrage_in_G}.

\subsection{Market viability under inside information}	\label{sec:arbitrage_in_G}

An agent endowed with inside information ({\em informed agent}) is supposed to have access to the information generated by $L$, i.e., to the initially enlarged filtration $\bG$. An informed agent can trade in the same set of securities available in the ordinary financial market, but is allowed to rely on the information flow $\bG$ when constructing portfolios. We call the tuple $(\Omega,\bG,\prob;S)$ the {\em insider financial market}, recalling that Assumption \ref{ass:Jacod} ensures that $S$ is a semimartingale on $\basisgp$ (see Lemma \ref{lem:prelim_G}).

We are especially interested in the situation where the additional information generated by $L$ yields arbitrage opportunities, so that NFLVR does not hold in the insider financial market $(\Omega,\bG,\prob;S)$. However, we need to ensure that $(\Omega,\bG,\prob;S)$ still represents a viable financial market, in the sense that portfolio optimization problems can be meaningfully solved. To this effect, in view of \cite[Proposition 4.19]{MR2335830}, the minimal no-arbitrage requirement is represented by the {\em no unbounded profit with bounded risk} (NUPBR) condition, defined as the boundedness in probability of the set $\bigl\{(\varphi\cdot S)_T : \varphi\in L(S,\bG)\text{ and }(\varphi\cdot S)_t\geq-1 \text{ a.s. for all $t\in[0,T]$}\bigr\}$.
By \cite[Theorem 2.6]{TakSch} (see also \cite[Theorem 2.1]{MR2972237} in the case $d=1$), $S$ satisfies NUPBR on $\basisgp$ if and only if 
\[
\Z := \bigl\{ Z\in\Mloc(\prob,\bG) \colon  Z>0, \; Z_0 =1 \text{ and }  ZS\in\Mloc(\prob,\bG) \bigr\} 
\neq \emptyset,
\]
with $\Z$ denoting the set of {\em equivalent local martingale deflators} (ELMDs) for $S$ on $\basisgp$.

The following result provides a complete characterization of the (no-)arbitrage properties of the insider financial market $(\Omega,\bG,\prob;S)$, in the sense of both NUPBR and NFLVR.

\begin{thm}	\label{thm:results_NUPBR}
Suppose that the space $L^1(\Omega,\F_T,\prob)$ is separable. 
Then, NUPBR holds on $\basisgp$ if and only if the set $\{q^x=0<q^x_-\}$ is evanescent for $\lambda$-a.e. $x\in E$. In this case, it holds that
$\Z=\{Z/q^L\}$.
Moreover, the following properties are equivalent:
\begin{enumerate}
\item[(i)] $S$ satisfies NFLVR on $\basisgp$;
\item[(ii)] for all $t\in[0,T]$, $\lambda\ll\nu_t$ holds in the a.s. sense;
\item[(iii)] $\prob(q^x_T>0)=1$ for $\lambda$-a.e. $x\in E$;
\item[(iv)] $\expec[1/q^L_T]=1$;
\item[(v)] $\expec[Z_T/q^L_T]=1$;
\item[(vi)] the process $1/q^L=(1/q^L_t)_{t\in[0,T]}$ is a martingale on $\basisgp$;
\item[(vii)] the process $N/q^L=(N_t/q^L_t)_{t\in[0,T]}$ is a martingale on $\basisgp$, for every  $N\in\M(\prob,\bF)$.
\end{enumerate}
\end{thm}

The sufficiency of the condition that $\{q^x=0<q^x_-\}$ is evanescent for $\lambda$-a.e. $x\in E$ for NUPBR to hold in $\bG$ has been shown in \cite[Theorem 1.12]{AFK} (see also \cite[Theorem 6]{ACJ15}). In our setting, the completeness of the financial market enables us to prove that the same condition is also necessary for NUPBR to hold in $(\Omega,\bG,\prob;S)$.
We point out that, in Theorem \ref{thm:results_NUPBR}, the separability assumption is only needed in the proof of the necessity part. 
Motivated by the above theorem, we now introduce our last standing assumption, which will be assumed throughout the paper. 

\begin{stass}	\label{ass:NoJumpZero}
The set $\{q^x=0<q_{-}^x\}$ is evanescent for $\lambda$-a.e. $x\in E$.
\end{stass}

We are especially interested in the case where the densities $q^x$ can reach zero, as this is intimately related to the existence of arbitrage opportunities in the insider financial market $(\Omega,\bG,\prob;S)$. In general, the densities $q^x$ can reach zero either in a continuous way or due to a jump to zero. Assumption \ref{ass:NoJumpZero} excludes a jump-to-zero behavior.
As shown in Theorem \ref{thm:results_NUPBR}, under our standing assumptions, the set of ELMDs for $S$ on $\basisgp$ is non-empty and consists of a singleton. 
In particular, in view of \cite[Corollary 2.1]{SY98}, the latter property implies that the insider financial market $(\Omega,\bG,\prob;S)$ inherits the completeness of the ordinary financial market $(\Omega,\bF,\prob;S)$.

The last part of Theorem \ref{thm:results_NUPBR} gives a precise characterization of when the inside information generates arbitrage opportunities for an informed agent. This happens if and only if the $\F_T$-conditional law $\nu_T$ of $L$ fails to be  equivalent with respect to the unconditional law $\lambda$. 
The failure of the equivalence means that there exist some scenarios that, from the point of view of an ordinary agent, are a priori possible (i.e., they have a strictly positive $\lambda$-measure) but can be later revealed to be impossible (i.e., they can be assigned zero $\nu_T$-measure). For an informed agent, such scenarios would be excluded already before the beginning of trading, thus providing a clear informational advantage. 
This phenomenon will also be clarified by the examples considered in Section \ref{sec:examples}.


\begin{rem}[On optimal arbitrage]	\label{rem:opt_arb}
Theorem \ref{thm:results_NUPBR} shows that NFLVR holds in $(\Omega,\bG,\prob;S)$ if and only if $\expec[Z_T/q^L_T]=1$. 
Observe that $\expec[Z_T/q^L_T]$ corresponds to the average cost of replicating the constant payoff $1$ in $(\Omega,\bG,\prob;S)$, as shown in the proof of Theorem \ref{thm:results_NUPBR}.
In the terminology of \cite{CT15}, the quantity $1/\expec[Z_T/q^L_T]$ is the {\em optimal arbitrage profit}, if $\expec[Z_T/q^L_T]<1$. In this sense, Theorem \ref{thm:results_NUPBR} shows that NFLVR holds in $(\Omega,\bG,\prob;S)$ if and only if no optimal arbitrage is possible.
\end{rem}

\begin{rem}[On the num\'eraire portfolio]	\label{rem:num_port_G}
The unique ELMD $Z/q^L$ is {\em tradable}, in the sense that there exists a process $\phi\in L(S,\bG)$ such that $q^L/Z=1+\phi\cdot S$ (this is a direct consequence of Proposition \ref{prop:MRT}, taking $M\equiv 1$). In other words, adopting the terminology of \cite{KKS}, the process $1+\phi\cdot S$ is the {\em local martingale num\'eraire} for the insider financial market $(\Omega,\bG,\prob;S)$.
\end{rem}

\section{Optimal consumption-investment problems under inside information}\label{sec:portfolio}

In this section, we study general optimal consumption-investment problems via duality techniques, allowing for state-dependent utilities and intermediate consumption.
Similarly to \cite{ABS}, we allow for a non-trivial initial information,  represented by the inside information generated by the random variable $L$.
However, since we are especially interested in the case where the knowledge of $L$ yields arbitrage opportunities, we have to depart from a classical duality approach based on martingale measures.
The results of this section represent fundamental ingredients for the computation of the value of informational arbitrage, which will be discussed in Section \ref{sec:indiff_price}.

\subsection{Admissible portfolios}	\label{sec:strategies}

We fix a {\em stochastic clock} $\kappa=(\kappa_t)_{t\in[0,T]}$, which is a non-decreasing c\`adl\`ag $\bF$-adapted bounded process with $\kappa_0\!=\!0$ and such that $\prob(\kappa_T\!>\!0|\sigma(L))>0$ a.s. The stochastic clock $\kappa$ represents the notion of time according to which consumption is assumed to occur. 


A \emph{portfolio} is defined as a triplet $\Pi=(v,\vartheta,c)$, where $v\in\R$ represents an initial capital, $\vartheta=(\vartheta_t)_{t\in[0,T]}$ is an $\R^d$-valued $S$-integrable process representing the holdings in the $d$ risky assets and $c=(c_t)_{t\in[0,T]}$ is a non-negative process representing the consumption rate. 
For an ordinary agent, the strategy $\vartheta$ and the consumption process $c$ are required to be measurable with respect to $\cP(\bF)$ and $\cO(\bF)$, respectively. On the other hand, an informed agent is allowed to construct portfolios by choosing $\cP(\bG)$-measurable strategies $\vartheta$ and $\cO(\bG)$-measurable consumption processes $c$.
The value process $V^{v,\vartheta,c}=(V^{v,\vartheta,c}_t)_{t\in[0,T]}$ of a portfolio $\Pi=(v,\vartheta,c)$ is defined as
\[
V^{v,\vartheta, c}_t := v + \int_0^t\vartheta_u\ud S_u - \int_0^tc_u \ud\kappa_u,
\qquad \text{ for all } t\in[0,T].
\]

\begin{defn}	\label{def:adm_strategies}
Let $\mathbf{H}\in\{\bF,\bG\}$, $k\in\R_+$ and $v\in\R$.
The set of {\em $\bH$-admissible portfolios} with initial capital $v$ and allowable credit line $k$, denoted by $\A^{\mathbf{H},k}_+(v)$, is defined as
\[
\A^{\mathbf{H},k}_+(v) :=
\left\{(\vartheta,c) \in L(S,\mathbf{H})\times\cO_+(\mathbf{H}) \;\colon\;
V_t^{v,\vartheta,c}\geq-k\text{ a.s. for all }t\in[0,T]\text{ and }V^{v,\vartheta,c}_T\geq0\text{ a.s.}\right\}.
\]
\end{defn}

According to Definition \ref{def:adm_strategies}, we assume that  investors have access to a finite and fixed credit line $k$ over the investment horizon $[0,T]$, and are required to fully repay their debts by the terminal date $T$. Observe that, in the absence of arbitrage opportunities, the requirement $V^{v,\vartheta,c}_T\geq0$ a.s. automatically implies that $V^{v,\vartheta,c}_t\geq0$ a.s. for all $t\leq T$ (see \cite[Proposition 3.5]{DS94}), so that $\A^{\mathbf{H},k}_+(v)=\A^{\mathbf{H},0}_+(v)$, for all $k\in\R_+$. However, this is no longer true in the presence of arbitrage opportunities.
For $k=0$, we recover from Definition \ref{def:adm_strategies}  the usual notion of admissibility via non-negative portfolios, as considered for instance in \cite{KS98,GP98,Amen,GK03,Mos15,CCFM2015}.

For convenience of notation, we define the processes $Z^{\bF}=(Z^{\bF}_t)_{t\in[0,T]}$ and $Z^{\bG}=(Z^{\bG}_t)_{t\in[0,T]}$ by
\[
Z^{\bF}_t := Z_t \qquad\text{ and }\qquad
Z^{\bG}_t := Z_t/q^L_t,
\qquad\text{ for all }t\in[0,T]. 
\]
For later use, let us also define the two following classes of portfolios:
\begin{align*}
\A_{sm}^{\mathbf{H},k}(v)  := \Biggl\{ (\vartheta, c)\in\, & \, L(S,\mathbf{H})\times\cO_+(\mathbf{H}) \;\colon\; \Biggr.
\int_0^T\!Z^{\mathbf{H}}_u\,c_u\ud\kappa_u\in L^1(\prob),
\quad V^{v,\vartheta,c}_T\geq0 \text{ a.s.}	\\
&\Biggl.\text{ and }
Z^{\mathbf{H}}\,V^{v+k,\vartheta,c} + \int_0^{\cdot} Z^{\mathbf{H}}_u\,c_u\ud\kappa_u
\text{ is a supermartingale on $(\Omega,\mathbf{H},\prob)$} 
\Biggr\},
\end{align*}
\[
\A^{\mathbf{H},k}_m(v) := 
\biggl\{(\vartheta,c)\in\A_{sm}^{\mathbf{H},k}(v) \;\colon\;
Z^{\mathbf{H}}\,V^{v+k,\vartheta,c} + \int_0^{\cdot} Z^{\mathbf{H}}_u\,c_u\ud\kappa_u\in\M(\prob,\mathbf{H})\biggr\}.
\]

As will be shown in the next subsection, the classes $\A_{sm}^{\mathbf{H},k}(v)$ and $\A^{\mathbf{H},k}_m(v)$ appear naturally in the solution of optimal investment-consumption problems.
The requirement $\int_0^TZ^{\mathbf{H}}_u\,c_u\ud\kappa_u\in L^1(\prob)$ ensures that the consumption process $c=(c_t)_{t\in[0,T]}$ can be financed via self-financing trading starting from some initial capital, in the sense that there exists a pair $(\zeta,\varphi)\in L^1_+(\prob,\mathcal{H}_0)\times L(S,\mathbf{H})$ such that $V_t^{\zeta,\varphi,c}\geq0$ a.s. for all $t\in[0,T]$, with $L^1_+(\prob,\mathcal{H}_0)$ denoting the set of $\mathcal{H}_0$-measurable integrable non-negative random variables (compare with Lemma \ref{lem:SuperMart=Mart}).

In the present setting characterized by two financial markets, a suitable definition of admissibility should ensure that every portfolio which is admissible for an ordinary agent is also admissible for an informed agent, in line with economic intuition. This property, as well as the relations between the sets $\A^{\mathbf{H},k}_+(v)$, $\A_{sm}^{\mathbf{H},k}(v)$ and $\A^{\mathbf{H},k}_m(v)$, is clarified in the next lemma.

\begin{lem}	\label{lem:AFinAG}
For every $v\in\R$ and $k\in\R_+$, the following holds:
\begin{enumerate}
\item[(i)] $\A^{\mathbf{H},k}_+(v)=\A_{sm}^{\mathbf{H},k}(v)$, for $\mathbf{H}\in\{\bF,\bG\}$;
\item[(ii)] $\A^{\bF,k}_+(v)\subseteq\A^{\bG,k}_+(v)$ and, therefore, $\A^{\bF,k}_{sm}(v)\subseteq\A^{\bG,k}_{sm}(v)$.
\end{enumerate}
Moreover, the inclusion $\A^{\bF,k}_m(v)\subseteq\A^{\bG,k}_m(v)$ holds for every $v\in\R$ if and only if $\expec[1/q^L_T]=1$.
\end{lem}

In view of Theorem \ref{thm:results_NUPBR}, the above lemma shows that the inclusion $\A_m^{\bF,k}(v)\subseteq\A_m^{\bG,k}(v)$ holds if and only if there are no arbitrage opportunities in $(\Omega,\bG,\prob;S)$. 
If we think of martingales as fair games, then the economic intuition underlying this property becomes clear. Indeed, portfolios which are fair for an ordinary agent can be considered too expensive (and, hence, unfair) by an informed agent if the latter has the possibility of achieving arbitrage.
This intuition is in line with the implication (i)$\Rightarrow$(v) in Theorem \ref{thm:results_NUPBR}, where it is shown that if $\expec[Z_T/q^L_T]<1$ (equivalently, $\expec[1/q^L_T]<1$), then investing the total wealth in the riskless asset is not a fair strategy for an informed agent, as the latter can replicate the constant payoff $1$ starting from initial capital $v<1$.

\subsection{Optimal consumption-investment problems}	\label{sec:general_utility}

We assume that preferences are defined with respect to intermediate consumption over $[0,T]$ and/or wealth at the terminal date $T$. More specifically, similarly as in \cite{Zit05,Mos15,CCFM2015}, we introduce a {\em utility stochastic field} $U=U(\omega,t,x):\Omega\times[0,T]\times\R_+\rightarrow\R\cup\{-\infty\}$ satisfying the following requirements.

\begin{ass}	\label{ass:U}
For every $(\omega,t)\in\Omega\times[0,T]$, the function $x\mapsto U(\omega,t,x)$ is strictly concave, strictly increasing, continuously differentiable on $(0,+\infty)$ and satisfies the Inada conditions
\[
\underset{x\downarrow0}{\lim} \, U'(\omega,t,x) = +\infty
\qquad\text{and}\qquad
\underset{x\rightarrow+\infty}{\lim} \, U'(\omega,t,x) = 0,
\]
with $U'$ denoting the derivative of $U$ with respect to $x$. By continuity, we assume that $U(\omega,t,0)=\lim_{x\downarrow0}U(\omega,t,x)$.
Finally, for every $x\geq0$, the stochastic process $U(\cdot,\cdot,x)$ is $\cO(\bF)$-measurable.
\end{ass}

In the following, we shall always assume that a utility stochastic field satisfies Assumption \ref{ass:U}, unless otherwise mentioned.
For $\mathbf{H}\in\{\bF,\bG\}$, we define the following set of consumption processes:
\begin{align*}
\Cons_+^{\mathbf{H},k}(v) := \left\{c\in\cO_+(\mathbf{H})\;\colon\;
\exists\; \vartheta\in L(S,\mathbf{H})\text{ s.t. }(\vartheta,c)\in\A_+^{\mathbf{H},k}(v)\right\},
\end{align*}
corresponding to all consumption plans that can be financed by portfolios with initial capital $v$ satisfying the allowable credit line $k$.
The optimal consumption-investment problem of an agent having access to the information flow $\mathbf{H}$ and with initial capital $v$ at  $t=0$ is  defined as follows:\footnote{For simplicity of notation, we shall omit to write explicitly the dependence on $\omega$ in the utility stochastic field $U$.}
\be	\label{eq:uH}
u^{\mathbf{H},k}(v) := \sup_{c\in\Cons_+^{\mathbf{H},k}(v)}
\expec\left[\int_0^TU(u,c_u)\ud\kappa_u \right],
\ee
with the convention $\expec[\int_0^TU(u,c_u)\ud\kappa_u]=-\infty$ if $\expec[\int_0^TU^-(u,c_u)\ud\kappa_u]=+\infty$. 
We also define the related $\mathcal{H}_0$-conditional optimization problem (with possibly non-trivial initial information):
\be	\label{eq:uH_2}
\underset{c\in\Cons_+^{\mathbf{H},k}(v)}{\esssup}\;
\expec\left[\int_0^TU(u,c_u)\ud\kappa_u \bigg|\mathcal{H}_0\right],
\ee
with an analogous convention.
Note that an element $c\in\Cons_+^{\mathbf{H},k}(v)$ attains the supremum in problem \eqref{eq:uH} if it attains the supremum in problem \eqref{eq:uH_2} (see, e.g., \cite[Section 4]{ABS}).
We also remark that the set $\Cons_+^{\mathbf{H},k}(v)$ is closed in the topology of convergence in measure $(\!\ud\kappa\otimes\prob)$ as long as NUPBR holds on $(\Omega,\mathbf{H},\prob;S)$ (see \cite{CCFM2015} and compare also with Lemma \ref{lem:SuperMart=Mart}).


The above structure of preferences is very general and allows for state-dependent utilities. By suitably specifying $\kappa$, several different formulations of optimal investment problems, with or without intermediate consumption, can be recovered from the present setting (compare with \cite[Section 2.8]{Zit05} and \cite[Examples 2.5-2.9]{Mos15}). In particular, the classical problem of maximization of expected utility from terminal wealth is obtained by setting $\ud\kappa_u=\delta_{T}(\!\ud u)$.


For convenience of notation, let us define the following sets of consumption processes, corresponding to the different sets of admissible portfolios considered in Section \ref{sec:strategies}:
\begin{align*}
\Cons_{sm}^{\mathbf{H},k}(v) &:= \left\{c\in\cO_+(\mathbf{H})\;\colon\;
\exists\; \vartheta\in L(S,\mathbf{H})\text{ s.t. }(\vartheta,c)\in\A_{sm}^{\mathbf{H},k}(v)\right\},\\
\Cons_m^{\mathbf{H},k}(v) &:= \left\{c\in\cO_+(\mathbf{H})\;\colon\;
\exists\; \vartheta\in L(S,\mathbf{H})\text{ s.t. }(\vartheta,c)\in\A_m^{\mathbf{H},k}(v) \text{ and }V^{v,\vartheta,c}_T=0\text{ a.s.}\right\}.
\end{align*}

It is evident that $\Cons_m^{\mathbf{H},k}(v)\subseteq\Cons_+^{\mathbf{H},k}(v)$, for every $k\in\R_+$ and $v\in\R$.
In the following lemma, we show that the optimal expected utility  does not change if one maximizes over all consumption processes belonging to the smaller class $\Cons_m^{\mathbf{H},k}(v)$. 
The economic intuition of this result is clear, as a martingale is the least expensive supermartingale that reaches a given terminal value (compare with \cite[Lemma 10.4.1]{MR2267213}). Furthermore, the condition $V^{v,\vartheta,c}_T=0$ a.s. represents the simple fact that all the available resources are used to finance consumption.

\begin{lem}	\label{lem:SuperMart=Mart}
Let $\mathbf{H}\in\{\bF,\bG\}$, $k\in\R_+$ and $v\in\R$. 
Then, $\Cons_+^{\mathbf{H},k}(v)=\Cons_{sm}^{\mathbf{H},k}(v)$ and, for every consumption process $c\in\cO_+(\mathbf{H})$, the following holds:
\begin{enumerate}
\item[(i)] $c\in\Cons_+^{\mathbf{H},k}(v)$ if and only if $\expec[\int_0^TZ^{\mathbf{H}}_u\,c_u\ud\kappa_u|\mathcal{H}_0]\leq v+k(1-\expec[Z^{\mathbf{H}}_T|\mathcal{H}_0])$ a.s.;
\item[(ii)] $c\in\Cons_m^{\mathbf{H},k}(v)$ if and only if $\expec[\int_0^TZ^{\mathbf{H}}_u\,c_u\ud\kappa_u|\mathcal{H}_0]=v+k(1-\expec[Z^{\mathbf{H}}_T|\mathcal{H}_0])$ a.s.
\end{enumerate}
Moreover, it holds that
\be	\label{eq:utility_mart}
u^{\mathbf{H},k}(v) = \sup_{c\in\Cons_m^{\mathbf{H},k}(v)}
\expec\left[\int_0^TU(u,c_u)\ud\kappa_u  \right]
=: u^{\mathbf{H},k}_m(v).
\ee
\end{lem}

\begin{rem}	\label{rem:credit_arbitrage}
It is important to observe that the credit line (or allowable leverage) $k$ plays no role in the characterization of financeable consumption plans if and only if $Z^{\mathbf{H}}\in\M(\prob,\mathbf{H})$. In view of \eqref{eq:utility_mart}, this implies that $u^{\mathbf{H},k}(v)=u^{\mathbf{H},0}(v)$, for every $v\in\R_+$, if and only if $Z^{\mathbf{H}}\in\M(\prob,\mathbf{H})$. In other words, the optimal expected utility does not depend on the allowable leverage if and only if there are no arbitrage opportunities in $(\Omega,\mathbf{H},\prob;S)$. To this effect, see also Remark \ref{rem:credit_cons}.
\end{rem}

We are now in a position to derive the general solution to the optimal consumption-investment problems \eqref{eq:uH}-\eqref{eq:uH_2}. Similarly as in \cite{Amen,ABS}, we rely on a  duality approach. However, since in our setting arbitrage opportunities can exist, we have to rely on ELMDs instead of martingale measures (compare with \cite[Chapter 3]{KS98} and \cite{FonRun}).
Due to the strict concavity and continuous differentiability of the utility stochastic field $U$ (see Assumption \ref{ass:U}), there exists a unique stochastic field $I=I(\omega,t,y):\Omega\times[0,T]\times(0,+\infty)\rightarrow(0,+\infty)$ such that $U'(\omega,t,I(\omega,t,y))=y$, for all $(\omega,t,y)\in\Omega\times[0,T]\times(0,+\infty)$.
Observe that, due to Assumption \ref{ass:U}, for every strictly positive $\mathbf{H}$-optional process $(Y_t)_{t\in[0,T]}$, it holds that $(I(\omega,t,Y_t(\omega)))_{t\in[0,T]}\in\cO_+(\mathbf{H})$.

\begin{prop}	\label{prop:lambda}
Let $\mathbf{H}\in\{\bF,\bG\}$, $k\in\R_+$ and $v\geq -k(1-\|\expec[Z^{\mathbf{H}}_T|\mathcal{H}_0]\|_{\infty})=:v^{\mathbf{H}}_k$. 
Suppose that there exists an $\mathcal{H}_0$-measurable random variable $\Lambda^{\mathbf{H},k}(v)\colon\Omega\rightarrow(0,+\infty)$ such that\footnote{For brevity of notation, we omit to write explicitly the dependence on $\omega$ in the stochastic field $I$.}
\be	\label{eq:condition_lambda}
\expec\bigg[ \int_0^TZ^{\mathbf{H}}_uI\left(u, \Lambda^{\mathbf{H},k}(v)Z^{\mathbf{H}}_u\right) \!\ud\kappa_u\bigg|\mathcal{H}_0\bigg] \ = \ v + k\left(1-\expec[Z^{\mathbf{H}}_T|\mathcal{H}_0]\right)
\quad\text{a.s.}
\ee
and such that the process $(I(t,\Lambda^{\mathbf{H},k}(v)Z^{\mathbf{H}}_t))_{t\in[0,T]}$ satisfies
$
\int_0^T\!U^-(u,I(u,\Lambda^{\mathbf{H},k}(v)Z^{\mathbf{H}}_u))\ud\kappa_u \in L^1(\prob)
$.
Then, the optimal consumption process $c^{\mathbf{H}}=(c^{\mathbf{H}}_t)_{t\in[0,T]}$ solving problem \eqref{eq:uH_2} with initial capital $v$ and allowable leverage $k$ is given by
$
c^{\mathbf{H}}_t = I(t,\Lambda^{\mathbf{H},k}(v)Z^{\mathbf{H}}_t) 
$,
for all $t\in[0,T]$.
\end{prop}

It can be easily shown that, if $u^{\mathbf{H},k}(v)<+\infty$, then the strict concavity of $U$ implies that the optimal consumption process $c^{\mathbf{H}}=(c^{\mathbf{H}}_t)_{t\in[0,T]}$ is unique up to a $(\!\ud\kappa\otimes\prob)$-nullset.
The associated optimal trading strategy $\vartheta^{\mathbf{H}}\in L(S,\mathbf{H})$ is given by the integrand appearing in the martingale representation (see Proposition \ref{prop:MRT}) of the local martingale $M=(M_t)_{t\in[0,T]}$ on $(\Omega,\mathbf{H},\prob)$ defined by $M_t:=\expec[\int_t^T\!Z^{\mathbf{H}}_uc^{\mathbf{H}}_u\ud\kappa_u|\mathcal{H}_t]+Z^{\mathbf{H}}_t\!\int_0^tc^{\mathbf{H}}_u\!\ud\kappa_u+k(\expec[Z^{\mathbf{H}}_T|\mathcal{H}_t]-Z^{\mathbf{H}}_t)$, for all $t\in[0,T]$.
Note also that the optimal solution does not depend on the allowable leverage $k$ if NFLVR holds on $(\Omega,\mathbf{H},\prob;S)$.

The quantity $v^{\mathbf{H}}_k=-k(1-\|\expec[Z^{\mathbf{H}}_T|\mathcal{H}_0]\|_{\infty})$ introduced in Proposition \ref{prop:lambda} admits a clear economic interpretation. Indeed, it represents the maximum amount of liabilities with which an agent can start at $t=0$. If $v=v^{\mathbf{H}}_k$, then an agent can start trading by borrowing the amount $k\|\expec[Z^{\mathbf{H}}_T|\mathcal{H}_0]\|_{\infty}$, thus exhausting his credit line, and investing in the self-financing strategy which replicates the constant payoff $k$. This strategy ensures the full repayment of all liabilities at date $T$ and requires an investment of $k\expec[Z^{\mathbf{H}}_T|\mathcal{H}_0]$ at $t=0$ (see the proof of Theorem \ref{thm:results_NUPBR}). The remaining resources $k(\|\expec[Z^{\mathbf{H}}_T|\mathcal{H}_0]\|_{\infty}-\expec[Z^{\mathbf{H}}_T|\mathcal{H}_0])$ can possibly be used to finance consumption.
For $v<v^{\mathbf{H}}_k$, there does not exist a strategy which can fully ensure the agent against his liabilities, so that $\Cons^{\mathbf{H},k}_+(v)=\emptyset$.

\begin{rem}	\label{rem:credit_cons}
Let $0\leq k_1<k_2$ and suppose that there exist two $\mathcal{H}_0$-measurable random variables $\Lambda^{\mathbf{H},k_1}(v)$ and $\Lambda^{\mathbf{H},k_2}(v)$ satisfying \eqref{eq:condition_lambda}, for some $v\geq v^{\mathbf{H}}_{k_1}$. Since $\prob(\kappa_T>0|\sigma(L))>0$ a.s., it can be shown that $\Lambda^{\mathbf{H},k_1}(v)\geq\Lambda^{\mathbf{H},k_2}(v)$ a.s., with strict inequality holding on $\{\expec[Z^{\mathbf{H}}_T|\mathcal{H}_0]<1\}$.
This means that, in the presence of arbitrage, a deeper credit line yields a higher consumption rate. In turn, this implies that $u^{\mathbf{H},k}(v)$ is strictly increasing in $k$ if $\expec[Z^{\mathbf{H}}_T]<1$ (see also Remark \ref{rem:credit_arbitrage}).
\end{rem}

\begin{rem}
The existence of an $\mathcal{H}_0$-measurable random variable $\Lambda^{\mathbf{H},k}(v)$ solving equation \eqref{eq:condition_lambda} is ensured if $\int_0^TZ^{\mathbf{H}}_uI(u,yZ^{\mathbf{H}}_u)\ud\kappa_u\in L^1(\prob)$, for all $y>0$. This corresponds to a classical condition in the theory of expected utility maximization (see, e.g., \cite{KLSX,KS98} and  \cite[Lemma 5.2]{Amen}).
\end{rem}


\begin{rem}	\label{rem:density_clock}
The structure of problems \eqref{eq:uH}-\eqref{eq:uH_2} together with the results of the present section make clear that, if there exists an $\bF$-stopping time $\tau$ such that $\kappa_{\tau}=\kappa_T$ a.s., then Assumptions \ref{ass:Jacod}-\ref{ass:NoJumpZero} can be relaxed and only assumed to hold on  $\dbra{0,\tau}$.
In particular, this allows to consider situations where Assumption \ref{ass:Jacod} holds on $[0,T)$ but fails at the terminal date $T$ (for instance when $L$ is an $\F_T$-measurable continuous random variable, in which case the problem of maximizing expected utility from wealth at date $T$ does not have a solution in $\bG$, see \cite{PK96,GP98,AIS98}).
\end{rem}

\subsection{Explicit solutions}	\label{sec:explicit_solutions}

In this section, we derive explicit solutions to the optimal consumption-investment problem in the case of logarithmic, power, and exponential utility functions. Besides allowing for intermediate consumption, this section generalizes \cite[Corollary 4.7]{ABS} to the case where the inside information can generate arbitrage opportunities. 
The results of this section will be used in Section \ref{sec:indiff_price} for the explicit computation of the value of informational arbitrage.
Corollaries \ref{cor:log_pwr_utility} and \ref{cor:exp_utility} can have some interest on their own, as they provide explicit solutions to optimal portfolio problems in general complete financial markets admitting arbitrage opportunities, in the presence of non-trivial initial information and leverage.

\begin{cor}	\label{cor:log_pwr_utility}
Let $\mathbf{H}\in\{\bF,\bG\}$, $k\in\R_+$ and $v>v^{\mathbf{H}}_k$. 
The optimal expected utilities in problem \eqref{eq:uH} for logarithmic and power utility functions are explicitly given as follows:
\begin{enumerate}
\item[(i)]
Let $U(\omega,t,x)=\log(x)$, for all $(\omega,t,x)\in\Omega\times[0,T]\times(0,+\infty)$.\\
If $\int_0^T\log(1/Z^{\mathbf{H}}_u)\ud\kappa_u\in L^1(\prob)$, then
\be	\label{eq:opt_ut_log}
u^{\mathbf{H},k}(v)
= \expec\left[\log\left(v+k(1-\expec[Z^{\mathbf{H}}_T|\mathcal{H}_0])\right)\kappa_T\right]
- \expec\bigl[\log\bigl(\expec[\kappa_T|\mathcal{H}_0]\bigr)\kappa_T\bigr] 
+ \expec\left[\int_0^T\log\left(\frac{1}{Z^{\mathbf{H}}_u}\right)\ud\kappa_u\right].
\ee
\item[(ii)]
Let $U(\omega,t,x)=x^p/p$, for some $p\in(0,1)$, for all $(\omega,t,x)\in\Omega\times[0,T]\times(0,+\infty)$.\\
If $\expec[\int_0^T(Z^{\mathbf{H}}_u)^{p/(p-1)}\ud\kappa_u|\mathcal{H}_0]<+\infty$ a.s., then
\be	\label{eq:opt_ut_pwr}
u^{\mathbf{H},k}(v)
= \frac{1}{p} \expec\left[ 
\left(v+k(1-\expec[Z^{\mathbf{H}}_T|\mathcal{H}_0])\right)^p
\expec\left[ \int_0^T \left(Z^{\mathbf{H}}_u \right) ^{\frac{p}{p-1}}\ud\kappa_u  \bigg| \mathcal{H}_0 \right]^{1-p} \right]
\ee
and $u^{\mathbf{H},k}(v)<+\infty$ if $\expec[\int_0^T(Z^{\mathbf{H}}_u)^{\frac{p}{p-1}}\ud\kappa_u |\mathcal{H}_0]^{1-p} \in L^1(\prob)$. 
\end{enumerate}
\end{cor}

Observe that the optimal expected utilities do not depend on $k$ if and only if there are no arbitrage opportunities in $(\Omega,\mathbf{H},\prob;S)$, in line with Remark \ref{rem:credit_arbitrage}. On the other hand, in the presence of arbitrage the optimal expected utility is strictly increasing in $k$, reflecting the fact that higher levels of consumption can be financed by taking more leveraged positions in the arbitrage strategy.

\begin{rem}	\label{rem:utility_gain}
Consider the classical setting where $\ud\kappa_u=\delta_T(\!\ud u)$ and $U(\omega,t,x)=\log(x)$, corresponding to maximization of expected logarithmic utility from terminal wealth.
Suppose that $u^{\bG,k}(v)<+\infty$, for some $k\in\R_+$ and $v>0$. Corollary \ref{cor:log_pwr_utility} then implies that
\begin{align}
u^{\bG,k}(v) - u^{\bF,k}(v)
&= \expec\left[\log\bigl(v+k(1-\expec[Z^{\bG}_T|\mathcal{G}_0])\bigr)\right]
+ \expec\left[\log\left(1/Z^{\bG}_T\right)\right] 
- \log(v) -\expec\left[\log\left(1/Z^{\bF}_T\right)\right] \notag\\
&= 
\expec\left[\log\left(1+\frac{k}{v}\Bigl(1-\qprob(q^x_T>0)\bigr|_{x=L}\Bigr)\right)\right]
+ \expec\bigl[\log(q^L_T)\bigr],
\label{eq:utility_gain}
\end{align}
representing the  {\em utility gain} of an informed agent with allowable leverage $k$. 
This result generalizes \cite[Theorem 3.7]{AIS98}, where relation \eqref{eq:utility_gain} has been obtained under the additional assumptions that the densities $q^x$ are a.s. strictly positive and continuous (and $k=0$).
Note also that 
\[
u^{\bG,k}(v)-u^{\bF,k}(v)
\geq u^{\bG,0}(v) - u^{\bF,0}(v)
\geq -\log(\expec[1/q^L_T])\geq0.
\] 
By Theorem \ref{thm:results_NUPBR}, this shows that the utility gain of an informed agent is always strictly positive if the inside information represented by $L$ allows for arbitrage in $(\Omega,\bG,\prob;S)$.
Moreover, if $L$ is a discrete $\F_T$-measurable random variable, as considered in Sections \ref{sec:example_Brownian}-\ref{sec:example_Poisson}, and $k=0$, the utility gain of an informed agent equals the entropy of $L$, i.e., $\expec[\log(q^L_T)]=-\sum_{x\in E}\prob(L=x)\log(\prob(L=x))$.
\end{rem}

Let us now consider the case of exponential preferences. 
Even though exponential utility does not satisfy Assumption \ref{ass:U} (the Inada condition at $0$ fails), the optimal consumption process can  be characterized similarly as in Proposition \ref{prop:lambda}.
The next corollary can be seen as a semimartingale version of  \cite[Theorem 2.4]{CH89} (see also \cite[Theorem 3.2]{MW12} in a discrete-time setting).

\begin{cor}	\label{cor:exp_utility}
Let $\mathbf{H}\in\{\bF,\bG\}$, $k\in\R_+$ and $v> v^{\mathbf{H}}_k$. Suppose that $U(\omega,t,x)=-e^{-\alpha x}$, for all $(\omega,t,x)\in\Omega\times[0,T]\times\R_+$, for $\alpha>0$. 
Then the optimal expected utility in problem \eqref{eq:uH} is given by
\be	\label{eq:opt_exp_utility}
u^{\mathbf{H},k}(v) = 
-\frac{1}{\alpha}\expec\left[\int_0^T\bigl(\Lambda^{\mathbf{H},k}(v)Z^{\mathbf{H}}_u\wedge \alpha\bigr)\ud\kappa_u\right],
\ee
where the $\mathcal{H}_0$-measurable random variable $\Lambda^{\mathbf{H},k}(v)$ is the a.s. unique solution to the equation
\be	\label{eq:exp_utility_lambda}
\frac{1}{\alpha}\expec\left[\int_0^TZ^{\mathbf{H}}_u\left(\log\left(\frac{\alpha}{\Lambda^{\mathbf{H},k}(v)Z^{\mathbf{H}}_u}\right)\right)^+\!\ud\kappa_u\bigg|\mathcal{H}_0\right] = v+k\left(1-\expec[Z^{\mathbf{H}}_T|\mathcal{H}_0]\right).
\ee
\end{cor}


Equation \ref{eq:exp_utility_lambda} can be explicitly solved in some simple models. In particular, if $\ud\kappa_u=\delta_T(\!\ud u)$ and $k=0$, a sufficient condition is that $\log(Z^{\mathbf{H}}_T)\leq\expec[Z^{\mathbf{H}}_T\log(Z^{\mathbf{H}}_T)|\mathcal{H}_0]/\expec[Z^{\mathbf{H}}_T|\mathcal{H}_0]$ a.s. The latter condition is always satisfied if $\qprob=\prob$ and $L$ is a discrete $\F_T$-measurable random variable generating arbitrage opportunities in $(\Omega,\bG,\prob;S)$
(compare also with the examples given in Sections \ref{sec:example_Brownian}-\ref{sec:example_Poisson}).

\section{The utility indifference value of inside information}		\label{sec:indiff_price}

By relying on the results established in the previous section, we are now in a position to study and compute the value of an inside information which potentially enables an informed agent to achieve arbitrage opportunities. Inspired by \cite{ABS}, we introduce the following definition.

\begin{defn}	\label{def:uip}
For $k\in\R_+$ and $v>0$, the {\em utility indifference value of the inside information $L$} is defined as a solution $\pi=\pi^{U,k}(v)\in\R_+$ to the following equation:
\begin{equation}\label{eq:price_arb}
u^{\FF,k}(v) \ = \ u^{\GG,k}(v - \pi).
\end{equation}
\end{defn}

As explained in the introduction, the value $\pi^{U,k}(v)$ is such that an investor is indifferent between two alternatives: 
(i) invest optimally the total initial wealth $v$ on the basis of the publicly available information; 
(ii) acquire the inside information $L$ at the price $\pi^{U,k}(v)$ and invest optimally the residual wealth $v-\pi^{U,k}(v)$, possibly exploiting the arbitrages generated by the knowledge of $L$.
If the inside information $L$ allows an investor to achieve arbitrage, as considered in the examples of Section \ref{sec:examples}, then we call the quantity $\pi^{U,k}(v)$ the {\em indifference value of informational arbitrage}.

The utility indifference value $\pi^{U,k}(v)$ exists and is unique under natural assumptions on the optimal consumption-investment problem,  as long as the expected utility maximization problem of an informed agent is well-posed (i.e., it does not lead to infinite utility).

\begin{thm}	\label{thm:IndPrice}
Suppose that $u^{\bF,0}(v)>-\infty$, for every $v>0$, and that the assumptions of Proposition \ref{prop:lambda} are satisfied, for every $k\in\R_+$, $v>v^{\mathbf{H}}_k$ and $\mathbf{H}\in\{\bF,\bG\}$.
Assume furthermore that $u^{\bG,k}(v_0)<+\infty$, for some $v_0>v^{\bG}_k$.
Then, for every $v>0$, the following hold:
\begin{enumerate}
\item[(i)]
If $\lim_{w\searrow v^{\bG}_k}u^{\bG,k}(w)\!<\!u^{\bF,0}(v)$, then the utility indifference value $\pi^{U,k}(v)$ exists and is unique.
\item[(ii)]
The map $k\mapsto\pi^{U,k}(v)$ is strictly increasing if and only if $\expec[1/q^L_T]<1$.
\item[(iii)]
If $\int_E(\expec[\int_0^T\!\ind_{\{q^x_t=0\}}\!\ud\kappa_t]+k\prob(q^x_T=0))\lambda(\!\ud x)>0$, then it always holds that $\pi^{U,k}(v)>0$.
\end{enumerate}
\end{thm}

The condition appearing in part (i) of the above theorem is always satisfied in the absence of leverage (i.e., if $k=0$). 
Indeed, $\Lambda^{\bG,0}(0):=\lim_{w\searrow0}\Lambda^{\bG,0}(w)$ exists by monotonicity (see the proof of Theorem \ref{thm:IndPrice} given in Section \ref{sec:proofs_3}), with values in $(0,+\infty]$.
By continuity of $I(\omega,t,\cdot)$ and dominated convergence, condition \eqref{eq:condition_lambda} gives that
\[
\expec\left[\int_0^TZ^{\bG}_uI\bigl(u,\Lambda^{\bG,0}(0)Z^{\bG}_u\bigr)\!\ud\kappa_u\biggr|\G_0\right] = 0.
\]
Therefore, recalling that $I(\omega,t,+\infty):=\lim_{y\rightarrow+\infty}I(\omega,t,y)=0$, for all $(\omega,t)\in\Omega\times[0,T]$, as a consequence of the Inada conditions (see Assumption \ref{ass:U}), it must hold that $\Lambda^{\bG,0}(0)=+\infty$ a.s. 
By Proposition \ref{prop:lambda} and the reverse Fatou lemma, this implies that, for every $v>0$,
\begin{align*}
\lim_{w\searrow0}u^{\bG,0}(w)
&= \lim_{w\searrow0}\expec\left[\int_0^TU\bigl(u,I(u,\Lambda^{\bG,0}(w)Z^{\bG}_u)\bigr)\!\ud\kappa_u\right]	\\
&\leq \expec\left[\int_0^TU\bigl(u,I(u,\Lambda^{\bG,0}(0)Z^{\bG}_u)\bigr)\!\ud\kappa_u\right]	\\
&= \expec\left[\int_0^TU\bigl(u,0\bigr)\ud\kappa_u\right]	
< \expec\left[\int_0^TU\Bigl(u,\frac{v}{\expec[\kappa_T]Z^{\bF}_u}\Bigr)\ud\kappa_u\right]
\leq u^{\bF,0}(v),
\end{align*}
where we have used the fact that the consumption process $(v/(\expec[\kappa_T]Z^{\bF}_t))_{t\in[0,T]}$ is strictly positive and belongs to $\Cons^{\bF,0}_+(v)$ (see Lemma \ref{lem:SuperMart=Mart}).
On the other hand, if $k>0$, then the condition $\lim_{w\searrow v^{\bG}_k}u^{\bG,k}(w)<u^{\bF,0}(v)$ may not necessarily hold, as the combined possibility of leverage and arbitrage may lead an informed agent to always outperform an uninformed agent, even when starting from a liability position at $t=0$.
In this situation, an agent with initial wealth $v$ would strictly prefer to acquire the inside information at any price not greater than $v-v^{\bG}_k$.
Note also that the assumption $u^{\bF,0}(v)>-\infty$, for every $v>0$, always holds if the utility stochastic field $U$ is bounded from below by a real-valued function (in particular, if $U$ is deterministic).

One of the implications of Theorem \ref{thm:IndPrice} is that, whenever the inside information $L$ yields arbitrage, then the indifference value of informational arbitrage is strictly increasing in the allowable leverage $k$. The economic intuition is that, having access to a deeper line of credit, an informed agent can  take  more leveraged positions in the arbitrage strategies, yielding arbitrage profits which can be scaled up to the limit of the allowable leverage.
In turn, this enables an informed agent to finance higher levels of consumption. Note that, if $k>0$, then it may happen that $\pi^{U,k}(v)>v$.

The condition $\int_E(\expec[\int_0^T\!\ind_{\{q^x_t=0\}}\!\ud\kappa_t]+k\prob(q^x_T=0))\lambda(\!\ud x)>0$ implies that an informed agent can finance any consumption plan $c\in\Cons_m^{\bF,k}(v)$ at a cost smaller than $v$, using the remaining resources to increase consumption. 
This is possible since an informed agent does not need to finance consumption in the states which are incompatible with the observed realization of $L$.
In particular, suppose that $\prob(\Delta\kappa_T>0)=1$ (or $k>0$) and that $\prob(q^x_T=0)>0$, for all $x$ belonging to some set $B$ with $\lambda(B)>0$. This corresponds to the situation where a strictly positive weight is placed on wealth at the terminal date $T$ and the random variable $L$ generates arbitrage (see Theorem \ref{thm:results_NUPBR}).
In this case, part (iii) of Theorem \ref{thm:IndPrice} implies that the indifference value of informational arbitrage is always strictly positive: an investor will always be willing to pay a strictly positive price to learn the inside information, regardless of the specific preference structure. 

The conclusions of Theorem \ref{thm:IndPrice} always hold for the utility functions considered in Section \ref{sec:explicit_solutions}, under suitable integrability conditions. This enables us to obtain explicit expressions for the utility indifference value of the inside information $L$ in the case $k=0$ for logarithmic and power utility functions, as shown in the next proposition, which generalizes \cite[Theorem 5.3]{ABS} to the case of an inside information yielding arbitrage opportunities and intermediate consumption.

\begin{prop}	\label{prop:pi}
Suppose that $k=0$. Then the utility indifference value of the inside information $L$ is explicitly given as follows:
\begin{enumerate}
\item[(i)]
Let $U(\omega,t,x)=\log(x)$, for all $(\omega,t,x)\in\Omega\times[0,T]\times(0,+\infty)$.\\
If $\int_0^T\log(q^L_u/Z_u)\ud\kappa_u\in L^1(\prob)$, then, for every $v>0$,
\be	\label{eq:ind_price_log}
\pi^{\rm log}(v) = v\left(1-\exp\left(\frac{1}{\expec[\kappa_T]}\left(\chi^{\bG}-\chi^{\bF}-\expec\left[\int_0^T\log(q^L_u)\ud\kappa_u\right]\right)\right)\right),
\ee
where $\chi^{\mathbf{H}}:=\expec[\log(\expec[\kappa_T|\mathcal{H}_0])\kappa_T]$, for $\mathbf{H}\in\{\bF,\bG\}$.
\item[(ii)]
Let $U(\omega,t,x)=x^p/p$, for some $p\in(0,1)$, for all $(\omega,t,x)\in\Omega\times[0,T]\times(0,+\infty)$.\\
If $\expec[\int_0^T(Z_u/q^L_u)^{\frac{p}{p-1}}\ud\kappa_u|\sigma(L)]^{1-p}\in L^1(\prob)$, then, for every $v>0$,
\be	\label{eq:ind_price_pwr}
\pi^{\rm pwr}(v) = v\left(1-\frac{\expec\biggl[\int_0^TZ_u^{\frac{p}{p-1}}\ud\kappa_u\biggr]^{\frac{1-p}{p}}}{\expec\Biggl[\expec\biggl[\int_0^T(Z_u/q^L_u)^{\frac{p}{p-1}}\ud\kappa_u\Bigr|\G_0\biggr]^{1-p}\Biggr]^{1/p}}\right).
\ee
\end{enumerate}
\end{prop}

In general, for $k>0$ the utility indifference value of the inside information cannot be computed in an explicit form for logarithmic and power preferences, as can be seen from \eqref{eq:opt_ut_log}-\eqref{eq:opt_ut_pwr}.\footnote{In view of Corollary \ref{cor:log_pwr_utility}, a fully explicit representation of the utility indifference value can be obtained when the random variable $k\,\expec[Z_T/q^L_T|\G_0]$ is a.s. constant or, equivalently, when $k\,\qprob(q^x_T>0)$ does not depend on $x$.} However, by part (ii) of Theorem \ref{thm:IndPrice}, formulae \eqref{eq:ind_price_log} and \eqref{eq:ind_price_pwr} represent lower bounds for the indifference value of the inside information for $k>0$ for logarithmic and power utility functions, respectively.

Proposition \ref{prop:pi} yields several interesting results on the value of informational arbitrage in the case of logarithmic and power utility functions. In particular:
\begin{itemize}
\item 
If $\int_E\expec[\int_0^T\!\ind_{\{q^x_t=0\}}\!\ud\kappa_t]\lambda(\!\ud x)>0$, then $\pi^{{\rm log}}(v)$ and $\pi^{{\rm pwr}}(v)$ are always strictly increasing with respect to $v$. This property is a direct consequence of Theorem \ref{thm:IndPrice} and formulae \eqref{eq:ind_price_log}-\eqref{eq:ind_price_pwr}. In other words, the value of informational arbitrage is strictly increasing with respect to initial wealth, in line with the analysis of \cite{Liu_et_al} in the case of a CRRA utility.
\item
In the case of logarithmic utility, the indifference value $\pi^{\log}(v)$ is lower when preferences are defined over intermediate consumption rather than terminal wealth only, confirming the empirical findings of \cite{Liu_et_al}.
This follows from the observation that
\[
\expec\left[\int_0^T\log(q^L_u)\!\ud\kappa_u\right]
\leq \expec\left[\int_0^T\log(q^L_T)\!\ud\kappa_u\right]
= \expec\left[\log(q^L_T)\kappa_T\right],
\]
where the inequality follows by taking the $\bG$-optional projection of $\log(q^L_T)$ and noting that, by Jensen's inequality and the $\bG$-supermartingale property of $1/q^L$,
\[
\expec\bigl[\log(q^L_T)|\G_t\bigr]
\geq \log\left(\expec[1/q^L_T|\G_t]^{-1}\right)
\geq \log(q^L_t)\;\text{ a.s. }
\qquad\text{ for all }t\in[0,T].
\]


\item
Jensen's inequality applied to the convex function $x\mapsto x\log x$ implies that the term $\chi^{\bG}-\chi^{\bF}$ appearing in \eqref{eq:ind_price_log} is non-negative, with $\chi^{\bG}=\chi^{\bF}$ if and only if $\expec[\kappa_T|\sigma(L)]=\expec[\kappa_T]$ a.s.
In turn, this means that if the inside information $L$ has predictive power on $\kappa_T$, then the indifference value $\pi^{{\rm log}}(v)$ is lower than in the case where $L$ has no predictive power on $\kappa_T$. 
While this result seems counterintuitive at first sight, it can be explained by the features of logarithmic preferences. 
We can think of $\kappa_T(\omega)$ as the total weight assigned to utility from consumption over $[0,T]$ in state $\omega$. By the specific structure of the optimal consumption process for logarithmic utility (see the proof of Corollary \ref{cor:log_pwr_utility}), if $\prob(\expec[\kappa_T|\sigma(L)]\neq\expec[\kappa_T])>0$, then an informed agent consumes more in the states $\omega$ which are more likely to be weighted less, and vice versa. 
However, as a consequence of risk aversion, an agent would a priori prefer to smooth consumption over different states. This intuitively explains why the presence of predictive power of $L$ on $\kappa_T$ leads to a lower indifference value of information.
Note that $\chi^{\bG}=\chi^{\bF}$ if $\kappa_T$ is deterministic, as in the case of utility from terminal wealth.


\end{itemize}


\begin{rem}
In the case of utility from terminal wealth (corresponding to $\ud\kappa_u=\delta_T(\!\ud u)$), it can be easily verified that formulae \eqref{eq:ind_price_log}-\eqref{eq:ind_price_pwr}  reduce to the expressions stated in \cite[Theorem 5.3]{ABS} whenever one of the equivalent conditions of the second part of Theorem \ref{thm:results_NUPBR} holds, i.e., whenever the inside information does not lead to arbitrage opportunities in $(\Omega,\bG,\prob;S)$.
For $\ud\kappa_u=\delta_T(\!\ud u)$, formula \eqref{eq:ind_price_log} reduces to $\pi^{\rm log}(v)=v(1-\exp(-\expec[\log(q^L_T)]))$. In line with Theorem \ref{thm:IndPrice} (see also Remark \ref{rem:utility_gain}), this confirms that the indifference value is always strictly positive if the inside information $L$ allows for arbitrage in $(\Omega,\bG,\prob;S)$. 
Moreover, the logarithmic indifference value is fully determined by the entropy of $L$ whenever $L$ is a discrete $\F_T$-measurable random variable.
\end{rem}

\subsection*{Universal results on the indifference value of inside information}

In general, the indifference value of the inside information depends on the stochastic utility field considered. However, in some special cases (for instance, in the example given in the introduction), the indifference value is a {\em universal} value, which does not depend on the preference structure. This situation is clarified by the next theorem.
We denote by $\mathcal{U}$ the class of all strictly increasing and concave deterministic utility functions $U:\R_+\rightarrow\R\cup\{-\infty\}$. In the statement of the following theorem, we denote by $u^{\mathbf{H},k}(v)$ the value function associated to problem \eqref{eq:uH} in the case of expected utility from consumption only at date $T$ (i.e., terminal wealth) with utility function $U$.

\begin{thm}	\label{thm:universal_uip}
Suppose that $\qprob=\prob$ in Assumption \ref{ass:Q} and that $\ud\kappa_u=\delta_T(\!\ud u)$.
Then, the following three conditions are equivalent:
\begin{enumerate}
\item[(i)] it holds that $\prob(q^L_T=q)=1$, for some constant $q\geq1$;
\item[(ii)]
for every $k\in\R_+$ and $v>0$, there exists a universal value $\pi^k(v)\in[0,v+k)$ such that
\[
u^{\bG,k}\bigl(v-\pi^k(v)\bigr) = u^{\bF,k}(v),
\qquad \text{ for all }U\in\mathcal{U};
\]
\item[(iii)] for every $v>0$, there exists a universal value $\pi^0(v)\in[0,v)$ such that
\[
u^{\bG,0}\bigl(v-\pi^0(v)\bigr) = u^{\bF,0}(v),
\qquad \text{ for all }U\in\mathcal{U}.
\]
\end{enumerate}
In those cases, for every $U\in\mathcal{U}$, $k\in\R_+$ and $v>0$, the indifference value $\pi^k(v)$ is always given by
\be	\label{eq:universal_uip}
\pi^k(v) = (v+k)\left(1-\frac{1}{q}\right)
\ee
and the optimal wealth process $V^{\bG}=(V^{\bG}_t)_{t\in[0,T]}$ in problem \eqref{eq:uH} for $\mathbf{H}=\bG$ is always given by
\be	\label{eq:universal_wealth}
V^{\bG}_t = (v+k)\frac{q^L_t}{q^L_T}-k,
\qquad\text{ for all }t\in[0,T].
\ee
\end{thm}

In the setting of the above theorem, the optimal strategy for an informed agent is given by a multiple of the process $\phi\in L(S,\bG)$ appearing in the stochastic integral representation $q^L=1+\phi\cdot S$ (see Remark \ref{rem:num_port_G}). This shows an interesting property: under the conditions of Theorem \ref{thm:universal_uip}, the optimal portfolio for an informed agent will always be the {\em num\'eraire portfolio} in the financial market $(\Omega,\bG,\prob;S)$, regardless of the preference structure.
Equivalently, the constant payoff $v=v-\pi^k(v)+(v+k-\pi^k(v))(\phi\cdot S)_T$ dominates according to the second order stochastic dominance criterion (see, e.g., \cite[Chapter 5]{Ingersoll}) all possible outcomes of admissible portfolios.

\begin{rem}	\label{rem:universal_discrete}
The random variable $q^L_T$ is always deterministic whenever $L$ is an $\F_T$-measurable discrete random variable with uniform distribution on a finite set $E$, so that $\prob(L=x)=1/|E|$ 
for all $x\in E$. Indeed, in this case it holds that $q^x_T=\ind_{\{L=x\}}|E|$, for all $x\in E$, so that $q^L_T=|E|$. This is also the case of Example \ref{ex:intro}, as we shall explain in detail in Section \ref{sec:example_Brownian}.
\end{rem}

\begin{rem}
If there are no arbitrage opportunities in $(\Omega,\bG,\prob;S)$, then the only case in which condition (i) of Theorem \ref{thm:universal_uip} holds is when the random variable $L$ is independent of $\F_T$.
Indeed, if NFLVR holds in $(\Omega,\bG,\prob;S)$ and  the random variable $q^L_T$ is a.s. constant, then $q^L_T=1$ a.s., as a consequence of Theorem \ref{thm:results_NUPBR}. Therefore, it holds that $q^x_T=1$ $(\prob\otimes\lambda)$-a.e., which by formula \eqref{expec_init} implies that $\expec[h(L)\ind_A]=\expec[h(L)]\prob(A)$, for every $\cE$-measurable bounded function $h:E\rightarrow\R$ and $A\in\F_T$, so that $L$ is independent of $\F_T$.
Conversely, if $L$ is independent of $\F_T$, then clearly $q^x_T=1$ for all $x\in E$.
In this case, formula \eqref{eq:universal_uip} implies that
it will never be attractive to buy the informational content of the random variable $L$, simply because the latter does not provide any useful information on the financial market.
\end{rem}

The assumptions of Theorem \ref{thm:universal_uip} cannot be easily relaxed. Indeed, if $\ud\kappa_u=\delta_T(\!\ud u)$ but $\qprob\neq\prob$, then condition (i) does not suffice to ensure the existence of a universal indifference value, as can be shown by a simple modification of the example given in Section \ref{sec:example_Brownian}.
Similarly, even if $\qprob=\prob$, in the presence of intermediate consumption the utility indifference value can depend on the preference structure even if $q^L_T$ is deterministic (apart from the trivial case where $L$ is independent of $\F_T$).

Under the same assumptions of Theorem \ref{thm:universal_uip}, we can establish some universal bounds for the indifference value of informational arbitrage, as shown in the following proposition.

\begin{prop}	\label{prop:uip_bounds}
Suppose that $\qprob=\prob$ in Assumption \ref{ass:Q} and that $\ud\kappa_u=\delta_T(\!\ud u)$.
Assume furthermore that there exist two strictly positive constants $q_{\min}$ and $q_{\max}$ with $q_{\min}\leq q_{\max}$ such that $\prob(q^L_T\in[q_{\min},q_{\max}])=1$.
Then, for every utility function $U\in\mathcal{U}$, $k\in\R_+$ and $v>0$, it holds that
\be	\label{eq:uip_bounds}
(v+k)\left(1-\frac{1}{q_{\min}}\right)^+ \leq \pi^{U,k}(v) \leq (v+k)\left(1-\frac{1}{q_{\max}}\right).
\ee
\end{prop}

The universal bounds derived in the above proposition will be illustrated in the context of the examples discussed in Sections \ref{sec:example_Brownian} and \ref{sec:example_Poisson}.

\section{Examples}\label{sec:examples}

In this section, we illustrate some of the main concepts and results  in the context of three examples. The first example (Section \ref{sec:example_Brownian}) consists of a generalization of Example \ref{ex:intro}. The second example (Section \ref{sec:example_Poisson}) considers a two-dimensional discontinuous financial market, where the inside information corresponds to the ratio of the terminal values of two assets. In these two examples, the random variable $L$ is discrete. In the third example (Section \ref{sec:example_cont}) we consider a continuous random variable $L$ generating informational arbitrage.

\subsection{One-dimensional geometric Brownian motion}	\label{sec:example_Brownian}

Let $W=(W_t)_{t\in[0,T]}$ be a one-dimensional Brownian motion on the filtered probability space $(\Omega,\mathcal{A},\bF,\prob)$, where $\bF=(\F_t)_{t\in[0,T]}$ is the $\prob$-augmentation of the natural filtration of $W$. We consider a financial market where a single risky asset is traded, with discounted price process $S=(S_t)_{t\in[0,T]}$ satisfying
\be	\label{eq:GBM}
\ud S_t= S_t\,\sigma_t\ud W_t, \qquad S_0 \ > \ 0,
\ee
where $\sigma=(\sigma_t)_{t\in[0,T]}$ is a strictly positive $\FF$-predictable process such that $\int_0^T\sigma_t^2\ud t<+\infty$ a.s.
According to the notation introduced in Section \ref{sec:ordinary}, the tuple $(\Omega,\bF,\prob;S)$ represents the ordinary financial market and Assumption \ref{ass:Q} is satisfied with $\qprob=\prob$.

Similarly as in \cite[Example 4.6]{PK96}, we suppose that the inside information is generated by the random variable $L:=\ind_{\{W_T\geq c\}}$, where $c$ is a constant such that $\prob(W_T \ge c) = r \in (0,1)$. In this setting, $E=\{0,1\}$ and the unconditional law of $L$ is given by $\lambda(\{0\})=1-r$ and $\lambda(\{1\})=r$. 
Since $L$ is discrete, Assumption \ref{ass:Jacod} is automatically satisfied. In particular, it holds that
\begin{align*}
q_t^0 
= \frac{\prob(L=0|\F_t)}{\prob(L=0)}
= \frac{1}{1-r}\,\Phi\bigg(\frac{c-W_t}{\sqrt{T-t}}\bigg), \qquad\qquad
q_t^1 
= \frac{\prob(L=1|\F_t)}{\prob(L=1)} 
= \frac{1}{r}\,\Phi\bigg(\frac{W_t-c}{\sqrt{T-t}}\bigg),
\end{align*}
for every $t\in[0,T)$, where $\Phi(x):=\int_{-\infty}^x \frac{1}{\sqrt{2\pi}}e^{-z^2/2}\ud z$.
For $t=T$, we have that
\[
q_T^0 
= \frac{1}{1-r}\,\ind_{\{W_T<c\}}, \qquad\qquad 
q_T^1 = \frac{1}{r}\,\ind_{\{W_T\geq c\}}.
\]
Since $q^0$ and $q^1$ have continuous paths, Assumption \ref{ass:NoJumpZero} is satisfied. 
Moreover, it holds that 
\[
q^L_T = \frac{1}{1-r}\,\ind_{\{W_T < c\}} + \frac{1}{r}\,\ind_{\{W_T \ge c\}}.
\]

In view of Theorem \ref{thm:results_NUPBR}, NUPBR holds in the insider financial market $(\Omega,\bG,\prob;S)$ and $1/q^L$ is the associated ELMD. However, since $\expec[1/q^L_T]<1$, the inside information leads to arbitrage and NFLVR does not hold.
The boundedness of $q^L_T$ ensures that all the assumptions of Proposition \ref{prop:pi} are satisfied and, therefore, we can compute explicitly the {\em indifference value of informational arbitrage}. For simplicity of presentation, let us consider the problem of maximizing expected utility of terminal wealth (i.e., $\ud\kappa_u=\delta_T(\!\ud u)$) for $k=0$. In this case, for every $v>0$, it holds that
\[
\pi^{\rm log}(v) = v\left(1-(1-r)^{1-r}r^{r}\right)
\qquad\text{and}\qquad
\pi^{\rm pwr}(v) = v\left(1-\bigl((1-r)^{1-p}+r^{1-p}\bigr)^{-1/p}\right).
\]
Observe that $\pi^{{\rm pwr}}(v)$ is increasing with respect to $p$, meaning that the indifference value of informational arbitrage is decreasing with respect to risk aversion. Furthermore, $\pi^{{\rm pwr}}(v)$ converges to $\pi^{{\rm log}}(v)$ as $p\rightarrow0$, for every $v>0$.

In the case of an exponential utility function with risk aversion $\alpha>0$, an application of Corollary \ref{cor:exp_utility} shows that
\[
u^{\bG,0}(v) = -\expec\bigl[e^{-\alpha vq^L_T}\bigr]
= -(1-r)e^{-\frac{\alpha v}{1-r}} - re^{-\frac{\alpha v}{r}},
\]
for every $v\in\R_+$. Therefore, the indifference value of informational arbitrage in the case of exponential utility is given by the unique solution $\pi=\pi^{\rm exp}(v)$ to the following  equation:
\[
e^{-\alpha v} = (1-r) e^{- \frac{\alpha}{1-r} (v - \pi)}+r e^{- \frac{\alpha}{r} (v - \pi) }.
\]

Note also that, in the context of the present example, for every strictly increasing and concave deterministic utility function $U:\R_+\rightarrow\R\cup\{-\infty\}$ and for every $k\in\R_+$, the indifference value of informational arbitrage $\pi^{U,k}(v)$ satisfies the following bounds, as a consequence of Proposition \ref{prop:uip_bounds}:
\[
\min\{r,1-r\} \leq \frac{\pi^{U,k}(v)}{v+k} \leq \max\{r,1-r\},
\qquad\text{ for all }v>0.
\]

\subsubsection*{Analysis of Example \ref{ex:intro}}

If $c=0$ (and, hence, $r=1/2$), the random variable $q^L_T$ reduces to  the constant $q^L_T=2$. In this case, in line with the result of Theorem \ref{thm:universal_uip} (see also Remark \ref{rem:universal_discrete}), the value of informational arbitrage for $k=0$ is equal to the universal value 
$
\pi^0(v)=v/2.
$
In view of formula \eqref{eq:universal_wealth}, the corresponding optimal wealth process $V^{\bG}=(V^{\bG}_t)_{t\in[0,T]}$ is given by
\[
V^{\bG}_t = v\frac{q^L_t}{q^L_T}
= v\left(\Phi\bigg(\frac{-W_t}{\sqrt{T-t}}\bigg)\ind_{\{W_T < 0\}} 
+  \Phi\bigg(\frac{W_t}{\sqrt{T-t}}\bigg)\ind_{\{W_T \ge 0\}}\right),
\]
for all $t\in[0,T]$.
An application of It\^o's formula yields $\Phi(\frac{W_t}{\sqrt{T-t}})=\frac{1}{2}+\frac{1}{\sqrt{2\pi}}\int_0^t\frac{1}{\sqrt{T-u}}\exp(-\frac{W^2_u}{2(T-u)})\!\ud W_u$, so that the optimal strategy $\vartheta^{\bG}=(\vartheta^{\bG}_t)_{t\in[0,T]}$ for the informed agent is explicitly given by
\be	\label{eq:universal_strategy}
\vartheta^{\bG}_t
= \left(\ind_{\{W_T\geq0\}}-\ind_{\{W_T<0\}}\right)\frac{v}{\sigma_t S_t}\frac{1}{\sqrt{2\pi(T-t)}}\exp\left(-\frac{W^2_t}{2(T-t)}\right),
\qquad\text{ for all }t\in[0,T),
\ee
regardless of the utility function being considered.
In particular, the strategy $\vartheta^{\bG}$ is an arbitrage strategy for an informed agent. Indeed, it holds that $(\vartheta^{\bG}\cdot S)_t=V^{\bG}_t-v/2>-v/2$, for all $t\in[0,T]$, and $(\vartheta^{\bG}\cdot S)_T=v/2>0$.
This shows that, by acquiring the inside information $L$ at the price $\pi^0(v)=v/2$ and following the strategy $\vartheta^{\bG}$, an informed agent can achieve exactly the terminal wealth $v$, which also corresponds to the optimal terminal wealth for an ordinary agent.

\begin{rem}[On the universal optimal strategy $\vartheta^{\bG}$]
In the setting of the present example, the optimal strategy $\vartheta^{\bG}$ calculated in \eqref{eq:universal_strategy} has several interesting features:
\begin{enumerate}
\item[(i)] 
the strategy is always long or short in the risky asset depending on the inside information revealed at the beginning of trading; 
\item[(ii)]
the strategy is a bet on the risky asset: the position on the risky asset increases if the asset price decreases and, vice versa, decreases if the asset price increases;
\item[(iii)] 
it holds that $\lim_{t\rightarrow T}\vartheta^{\bG}_t=0$, meaning that the position in the risky asset is completely liquidated at the end of the investment horizon; 
\item[(iv)] 
$\vartheta^{\bG}$ is a multiple of the trading strategy which realizes the optimal arbitrage in $\bG$ (see Remark \ref{rem:opt_arb} and compare with the proof of the implication (i)$\Rightarrow$(v) in Theorem \ref{thm:results_NUPBR}).
\end{enumerate}
\end{rem}

\subsection{Two-dimensional Poisson process}	\label{sec:example_Poisson}

The example of Section \ref{sec:example_Brownian} considers a single risky asset with continuous paths. We now present an example of inside information leading to arbitrage in a financial market with two risky assets with discontinuous paths. 

Let $N^1=(N^1_t)_{t\in[0,T]}$ and $N^2=(N^2_t)_{t\in[0,T]}$ be two independent Poisson processes with common intensity $1$ on a filtered probability space $(\Omega,\mathcal{A},\bF,\prob)$, where $\bF=(\F_t)_{t\in[0,T]}$ is the $\prob$-augmentation of the natural filtration of $(N^1,N^2)$. We consider two risky assets, with discounted price processes $S^1=(S^1_t)_{t\in[0,T]}$ and $S^2=(S^2_t)_{t\in[0,T]}$ satisfying
\begin{align*}
\ud S^1_t &= S^1_{t-}(\!\ud N^1_t -\ud t),
\qquad S^1_0>0;\\
\ud S^2_t &= S^2_{t-}(\!\ud N^2_t-\ud t),
\qquad S^2_0>0,
\end{align*}
with explicit solutions $S^i_t=S^i_0\,e^{-t}\,2^{N^i_t}$, for $i\in\{1,2\}$ and $t\in[0,T]$.
The tuple $(\Omega,\bF,\prob;(S^1,S^2))$ represents the ordinary financial market and, since $(S^1,S^2)$ has the martingale representation property on $\basisp$, Assumption \ref{ass:Q} is satisfied with $\qprob=\prob$.

Let us define the process $N=(N_t)_{t\in[0,T]}$ by $N_t:=N^1_t-N^2_t$, for all $t\in[0,T]$. We suppose that the inside information is generated by the observation of the random variable $L:=N_T$, corresponding to the knowledge of the ratio $S^1_T/S^2_T$ of the terminal prices of the two assets. 
The distribution of the random variable $L$ can been explicitly computed and is given by 
\[
\prob(L=x)
= e^{-2T}\mathcal{I}_{|x|}(2T)
= e^{-2T}\sum_{k\in\N}\frac{T^{2k+|x|}}{k!(k+|x|)!},
\qquad\text{ for all } x\in\mathbb{Z},
\] 
where $\mathcal{I}_{|x|}(2T)$ denotes the modified Bessel function of the first kind.
Since $L$ is discrete, Assumption \ref{ass:Jacod} is automatically satisfied and, similarly as in \cite{CRT}, it can be computed that
\[
q^x_t = \frac{\prob(L = x |\F_t)}{\prob(L = x)} 
= \frac{ \sum_{k\in\N} e^{-(T-t)} \frac{(T-t)^k}{k!} e^{-(T-t)} \frac{(T-t)^{k+x-N_t}}{(k+x-N_t)!} \ind_{\{k + x - N_t \ge 0\}}  }{ \sum_{k\in\N}e^{-2T}\frac{T^{2k+|x|}}{k!(k+|x|)!} }, 
\]
for all $x\in\mathbb{Z}$ and $t\in[0,T)$. For $t=T$, we have that
\[ 
q^x_T = \frac{\ind_{\{L = x\}}}{\prob(L = x)} 
= \frac{\ind_{\{L = x\}}}{ e^{-2T} \sum_{k\in\N} \frac{T^{2k+|x|}}{k!(k+|x|)!}},
\qquad\text{ for all }x\in\mathbb{Z}.
\]
Note that $q^x_t>0$, for all $t\in[0,T)$. Moreover, the filtration $\bF$ is quasi-left-continuous and, hence, it holds that $\bigvee_{t<T}\F_t=\F_{T-}=\F_T$. By the L\'evy upward theorem, as $t\rightarrow T$ it holds that 
\[
\expec[\ind_{\{L=x\}}|\F_t] \longrightarrow
\expec[\ind_{\{L=x\}}|\F_{T-}]
= \expec[\ind_{\{L=x\}}|\F_T]
= \ind_{\{L=x\}}\qquad\text{ a.s.}
\]
thus showing that $q^x$ does not jump to zero. Assumption \ref{ass:NoJumpZero} is therefore satisfied and the insider financial market $(\Omega,\bG,\prob;S)$ satisfies NUPBR (see Theorem \ref{thm:results_NUPBR}).
The inside information $L$ generates arbitrage opportunities for an informed agent, since $\expec[1/q^L_T]=\sum_{x\in\mathbb{Z}}\prob(L=x)^2<1$.
However, due to admissibility constraints, such arbitrage opportunities cannot be realized by naive long-and-short strategies in the two assets $S^1$ and $S^2$, as the latter are unbounded from above.

The indifference value of informational arbitrage can be explicitly computed in the case of logarithmic and power utility functions by relying on Proposition \ref{prop:pi} (for simplicity of presentation, we only consider the case $\ud\kappa_u=\delta_T(\!\ud u)$ and $k=0$):
\begin{gather*}
\pi^{\rm log}(v)
= v\left(1-\exp\biggl( - \sum_{x\in \mathbb{Z}} \prob(L=x)\log \prob(L=x)\biggr)  \right), 	\\
\pi^{\rm pwr}(v)
= v\left( 1- \expec\left[ \left( \sum_{x\in \mathbb{Z}} \ind_{\{L=x\}}\prob(L=x)^{p/(p-1)} \right) ^{1-p}\right] ^{-1/p} \right)
= v\left( 1-  \left( \sum_{x\in \mathbb{Z}} \prob(L=x)^{1-p} \right) ^{-1/p} \right),
\end{gather*}
for every $v>0$.
In particular, note that in the case of a logarithmic utility function the value of informational arbitrage is determined by the entropy of the random variable $L$ (see Remark \ref{rem:utility_gain}).

In the case of an exponential utility function with risk aversion $\alpha>0$, an application of Corollary \ref{cor:exp_utility} shows that $\pi^{\rm exp}(v)$ is given by the unique solution $\pi=\pi^{\rm exp}(v)$ to the following equation:
\[
e^{-\alpha v}  
= \expec\bigg[   \exp\Bigl( - \alpha q^L_T(v - \pi)  \Bigr)   \bigg]	
= \sum_{x \in \mathbb{Z}} \prob(L=x)e^{-\frac{\alpha(v-\pi)}{\prob(L=x)}}.
\]
Note also that, as a consequence of Proposition \ref{prop:uip_bounds}, for every strictly increasing and concave utility function $U$, the indifference value of informational arbitrage $\pi^{U,k}(v)$ is bounded from below by the quantity $(v+k)\prob(L\neq0)$, for every $v>0$. This follows from the fact that  $0=\argmax_{x\in\mathbb{Z}}\mathcal{I}_{|x|}(2T)$.

\subsection{Informational arbitrage induced by a continuous random variable}	\label{sec:example_cont}


The examples considered in Sections \ref{sec:example_Brownian}-\ref{sec:example_Poisson} involve discrete random variables. We now present an example of a filtration initially enlarged with respect to a continuous random variable $L$ satisfying the absolute continuity relation of Assumption \ref{ass:Jacod} and generating arbitrage opportunities for an informed agent.

Let $W=(W_t)_{t\in[0,T]}$ be a one-dimensional Brownian motion on $(\Omega,\A,\bF,\prob)$, where $\bF=(\F_t)_{t\in[0,T]}$ is the $\prob$-augmentation of the natural filtration of $W$. Let $U$ be a random variable with uniform distribution on $[0,1]$, independent of the Brownian motion $W$, and assume that $\A=\F_T\vee\sigma(U)$.
We consider a financial market with a single risky asset, with discounted price process $S=(S_t)_{t\in[0,T]}$ given as in \eqref{eq:GBM}.
Similarly as in Section \ref{sec:example_Brownian}, the tuple $(\Omega,\bF,\prob;S)$ represents the ordinary financial market and Assumption \ref{ass:Q} is satisfied with $\qprob=\prob$.

We define the random variable $L$ by
\[
L := \frac{W^*_T}{2(1+W^*_T)} + \frac{U}{1+W^*_T},
\]
where $W^*_T:=\sup_{t\in[0,T]}W_t$. 
The random variable $L$ takes values in $[0,1]$ and, conditionally on $\F_T$, is uniformly distributed on $[a(W^*_T),b(W^*_T)]$, where $a(y):=y/(2+2y)$ and $b(y):=(2+y)/(2+2y)$, for $y\in\R_+$.
Due to the reflection principle of Brownian motion, the unconditional law $\lambda$ of $L$ can be computed as
\[
\lambda([0,x]) = \prob(L\leq x)
= \expec\bigl[f(x,W^*_T)\bigr]
= \sqrt{\frac{2}{\pi T}}\int_0^{+\infty}f(x,z)\,e^{-\frac{z^2}{2T}}\ud z,
\]
for $x\in[0,1]$, where $f(x,z):=(z(x-1/2)+x)^+\wedge 1$, for all $(x,z)\in[0,1]\times\R_+$.
Similarly, for every $t<T$ and $x\in[0,1]$, it can be shown that (see, e.g., \cite[Exercise 3.1.6.7]{JYC09})
\begin{align*}
\nu_t([0,x]) = \prob(L\leq x|\F_t)
&= \expec\bigl[f(x,W^*_T)|\F_t\bigr]	\\
&= \sqrt{\frac{2}{\pi(T-t)}}\left(f(x,W^*_t)\int_0^{W^*_t-W_t}e^{-\frac{z^2}{2(T-t)}}\ud z + \int_{W^*_t}^{+\infty}f(x,z)\,e^{-\frac{(z-W_t)^2}{2(T-t)}}\ud z\right).
\end{align*}
For $z\in\R_+$, let us define the function $g(\cdot,z):[0,1]\rightarrow\R_+$ by $g(x,z):=(1+z)\ind_{[a(z),b(z)]}(x)$, for all $x\in[0,1]$. The $\F_t$-conditional density $q^x_t$ can then be expressed as
\[
q^x_t = \sqrt{\frac{T}{T-t}}\frac{g(x,W^*_t)\int_0^{W^*_t-W_t}e^{-\frac{z^2}{2(T-t)}}\ud z+ \int_{W^*_t}^{+\infty}g(x,z)e^{-\frac{(z-W_t)^2}{2(T-t)}}\ud z}{\int_0^{+\infty}g(x,z)\,e^{-\frac{z^2}{2T}}dz},
\qquad\text{ for all }x\in[0,1],
\]
for every $t\in[0,T)$ and, for $t=T$,
\[
q^x_T = \sqrt{\frac{\pi T}{2}}\frac{g(x,W^*_T)}{\int_0^{+\infty}g(x,z)e^{-\frac{z^2}{2T}}\ud z},
\qquad\text{ for all }x\in[0,1].
\]
The densities $q^x$ can be computed more explicitly by introducing the function $\gamma:[0,1]\setminus\{1/2\}\rightarrow\R_+$ given by $\gamma(x):=2x/(1-2x)$ for $x\in[0,1/2)$ and $\gamma(x):=(2-2x)/(2x-1)$ for $x\in(1/2,1]$.
Note that $g(x,z)=(1+z)\ind_{[0,\gamma(x)]}(z)$, for all $x\neq1/2$.
With this notation, for all $t\in[0,T)$, it holds that
\begin{align*}
q_t^x &= \ind_{\{W^*_t\leq\gamma(x)\}}\left(
\frac{(1+W^*_t)\sqrt{2\pi}\left(\Phi\left(\frac{W^*_t-W_t}{\sqrt{T-t}}\right)-\frac{1}{2}\right)}
{\sqrt{2\pi}\left(\Phi\left(\frac{\gamma(x)}{\sqrt{T}}\right)-\frac{1}{2}\right)
+ \sqrt{T}\Big(1 - e^{-\frac{\gamma(x)^2}{2T}}\Big)}\right.\\
&\left.\qquad\qquad
+\frac{(1+W_t)\sqrt{2\pi}\left(\Phi\left(\frac{\gamma(x)-W_t}{\sqrt{T-t}}\right)-\Phi\left(\frac{W^*_t-W_t}{\sqrt{T-t}}\right)\right)
+ \sqrt{T-t}\Big(e^{-\frac{(W^*_t-W_t)^2}{2(T-t)}} - e^{-\frac{(\gamma(x)-W_t)^2}{2(T-t)}}\Big)}
{\sqrt{2\pi}\left(\Phi\left(\frac{\gamma(x)}{\sqrt{T}}\right)-\frac{1}{2}\right)
+ \sqrt{T}\Big(1 - e^{-\frac{\gamma(x)^2}{2T}}\Big)}\right),
\end{align*}
for all $x\neq1/2$, while for $x=1/2$
\[
q_t^x = \frac{
\sqrt{2\pi}\left((W^*_t-W_t)\Phi\left(\frac{W^*_t-W_t}{\sqrt{T-t}}\right)+\frac{1}{2}+W_t-\frac{W^*_t}{2}\right)
+ \sqrt{T-t}\,e^{-\frac{(W^*_t-W_t)^2}{2(T-t)}}}
{\sqrt{\frac{\pi}{2}} + \sqrt{T}}.
\]
At the terminal date $t=T$, it holds that
\[
q^x_T = \ind_{\{W^*_T\leq\gamma(x)\}}
\frac{1+W^*_T}
{\left(2\Phi\left(\frac{\gamma(x)}{\sqrt{T}}\right)-1\right)
+ \sqrt{\frac{2T}{\pi}}\Big(1 - e^{-\frac{\gamma(x)^2}{2T}}\Big)},
\]
for all $x\neq1/2$, and, for $x=1/2$,
\[
q^x_T = \frac{1+W^*_T}{1 + \sqrt{\frac{2T}{\pi}}}.
\]
Therefore, we have that
\[
q^L_T = \frac{1+W^*_T}
{\left(2\Phi\left(\frac{\gamma(L)}{\sqrt{T}}\right)-1\right)
+ \sqrt{\frac{2T}{\pi}}\Big(1 - e^{-\frac{\gamma(L)^2}{2T}}\Big)}
\;\text{ a.s., }\quad\text{ with }\quad
\gamma(L) = \frac{1+W^*_T}{|1-2U|}-1.
\]


In this example, $\nu_t\ll\lambda$ holds a.s. for all $t\in[0,T]$, so that Assumption \ref{ass:Jacod} is satisfied.
However, $\nu_t$ and $\lambda$ fail to be equivalent, for every $t\in(0,T]$. This simply follows from the observation that $\nu_t$ is null outside of the interval $[a(W^*_t),b(W^*_t)]$, together with the fact that the process $(W^*_t)_{t\in[0,T]}$ is increasing and the functions $a(\cdot)$ and $b(\cdot)$ are increasing and decreasing, respectively.
Moreover, the continuity of the filtration $\bF$ implies that Assumption \ref{ass:NoJumpZero} is satisfied. In view of Theorem \ref{thm:results_NUPBR}, the insider financial market $(\Omega,\bG,\prob;S)$ satisfies NUPBR, but arbitrage opportunities do exist.
Note that the observation of the realization $L(\omega)$ corresponds to the knowledge that $W^*_T(\omega)\geq 2L(\omega)/(1-2L(\omega))$, if $L(\omega)<1/2$, or that $W^*_T(\omega)\leq 2(1-L(\omega))/(2L(\omega)-1)$, if $L(\omega)>1/2$. This information represents the informational advantage of an informed agent at time $t=0$.

The present example can be generalized to an absolutely continuous random variable $U$ with arbitrary cumulative distribution function $F_U$ and density $g_U$, independent of the Brownian motion $W$.
Let $H:\R\times\R_+\rightarrow\R$ be a function such that $H(u,z)=x$ if and only if $u=h(x,z)$, for all $(u,z)\in\R\times\R_+$, for some function $h:\R\times\R_+\rightarrow\R$ admitting the partial derivative $h_x:=\partial_x h$.
If we define the random variable $L:=H(U,W^*_T)$, it can be easily verified that the $\F_t$-conditional densities of $L$ can be obtained by the same computations as above, replacing $f(x,z)$ with $F_U(h(x,z))$ and $g(x,z)$ with $g_U(h(x,z))h_x(x,z)$.
 

\section{Proofs}		\label{sec:proofs}

\subsection{Proofs of the results of Section \ref{sec:setting}}
We start by presenting the proofs of the enlargement of filtration results stated in Sections \ref{sec:G} and \ref{sec:arbitrage_in_G}.

\begin{proof}[Proof of Lemma \ref{lem:prelim_G}]
The right-continuity of $\bG$ follows from \cite[Lemma 4.2]{Fon2015}, while the stability of the semimartingale property is established in \cite[Theorem 1.1]{Jac85}. The existence of the densities $\{q^x;x\in E\}$ can be proven similarly as in \cite[Lemma 1.8]{Jac85} and \cite[Lemma 2.2]{Amen} (see \cite[Lemma 2.3]{Fon2015} for details). The last part of the lemma follows from \cite[Corollary 1.11]{Jac85}.
\end{proof}

The following implication of Lemma~\ref{lem:prelim_G} will be constantly used: for every $t\in[0,T]$ and $(\cE\otimes\F_t)$-measurable function $E \times \Omega \ni (x, \omega)\mapsto f^x_t(\omega) \in \R_+$, it holds that
\be	\label{expec_init}
\expec\left[f_t^L\right] 
= \expec\left[\int_Ef_t^x\,q^x_t\,\lambda(\!\ud x)\right]
= \int_E\expec\left[f_t^x\,q^x_t\right]\lambda(\!\ud x).
\ee
As usual, formula \eqref{expec_init} can be extended to real-valued integrable $(\cE\otimes\F_t)$-measurable functions.

\begin{proof}[Proof of Proposition \ref{prop:MRT}]
Let us define the $\R^{d+1}$-valued semimartingale $X:=(1,S)$ and note that $ZX\in\Mloc(\prob,\bF)$. We first prove that $ZX$ has the martingale representation property on $\basisp$. To this effect, let $N=(N_t)_{t\in[0,T]}$ be a bounded martingale on $\basisp$ with $N_0=0$ such that $NZX\in\Mloc(\prob,\bF)$, meaning that $NZ\in\Mloc(\prob,\bF)$ and $NZS\in\Mloc(\prob,\bF)$. Since $\qprob\sim\prob$, it holds that $N\in\Mloc(\qprob,\bF)$ and $NS\in\Mloc(\qprob,\bF)$. In view of \cite[Theorem 11.3]{MR542115}, Assumption \ref{ass:Q} implies that $N$ is trivial $\qprob$-a.s. and, hence, $\prob$-a.s. Again by \cite[Theorem 11.3]{MR542115}, we conclude that $ZX$ has the martingale representation property on $\basisp$.
Let $M=(M_t)_{t\in[0,T]}\in\Mloc(\prob,\bG)$.
By \cite[Proposition 4.10]{Fon2015}, there exists a process $H\in L(ZX,\bG)$ such that
\be	\label{eq:proof_MRT}
M_t = \frac{1}{q^L_t}\bigl(M_0 + (H\cdot (ZX))_t\bigr)
= \frac{Z_t}{q^L_t}\frac{M_0 + (H\cdot (ZX))_t}{Z_t}
\qquad\text{ a.s. for all $t\in[0,T]$.}
\ee
Furthermore, due to the martingale representation property of $S$ on $\basisq$, there exists a process $\theta\in L(S,\bF)$ such that $1/Z=1+\theta\cdot S$. For each $n\in\N$, let us define $H^n:=H\ind_{\{\|H\|\leq n\}}$. Using integration by parts and the associativity of the stochastic integral, we have that
\[	\ba
\frac{M_0+H^n\cdot (ZX)}{Z}
&= M_0 + \bigl(M_0+(H^n\cdot (ZX))_-\bigr)\cdot \frac{1}{Z} + \frac{H^n}{Z_-}\cdot (ZX) + H^n\cdot \left[ZX,\frac{1}{Z}\right] \\
&= M_0 + \Bigl(\bigl(M_0+(H^n\cdot (ZX))_-\bigr)\theta\Bigr)\cdot S + H^n\cdot X-\bigl((H^n)^{\top}X_-Z_-\bigr)\cdot \frac{1}{Z}	\\
&= M_0 + K^n\cdot S,
\ea	\]
where the $\R^d$-valued process $K^n=(K^n_t)_{t\in[0,T]}$ is defined by 
\[
K^{n,i}_t:=(M_0+(H^n\cdot (ZX))_{t-}-(H^n)^{\top}_tX_{t-}Z_{t-})\theta^i_t+H^{n,i+1}_t,
\] 
for all $i=1,\ldots,d$ and $t\in[0,T]$. 
Arguing similarly as in \cite[Proposition 8]{RS97}, the fact that $H\in L(ZX,\bG)$ implies that $H^n\cdot (ZX)$ converges to $H\cdot (ZX)$ in the semimartingale topology as $n\rightarrow+\infty$. Hence, in view of \cite[Proposition III.6.26]{MR1943877}, $K^n\cdot S=(M_0+H^n\cdot (ZX))/Z-M_0$ also converges in the semimartingale topology to $K\cdot S$, for some $K\in L(S,\bG)$, thus proving that $(M_0+H\cdot(ZX))/Z=M_0+K\cdot S$. Together with \eqref{eq:proof_MRT}, this completes the proof. 
\end{proof}

We now prove our main result on the (no-)arbitrage properties of the insider financial market $(\Omega,\bG,\prob;S)$ (Theorem \ref{thm:results_NUPBR}).
As a preliminary, we recall the following result on the behavior of $\bF$-local martingales under an initial filtration enlargement. Under the standing Assumptions \ref{ass:Jacod}-\ref{ass:NoJumpZero}, this is a direct consequence of \cite[Proposition 9]{ACJ15} (see also \cite[Proposition 3.6]{AFK}).

\begin{lem}	\label{lem:mgG}
Let $M=(M_t)_{t\in[0,T]}$ be a local martingale on $\basisp$. Then $M/q^L\in\mMloc(\prob,\bG)$.
\end{lem}

\begin{proof}[Proof of Theorem \ref{thm:results_NUPBR}]
The sufficiency part of the first assertion of the theorem follows directly from \cite[Theorem 1.12]{AFK}. In order to prove the necessity, let use define the $\bF$-stopping times 
\[
\zeta^x:=\inf\{t\in[0,T]:q^x_t=0\}
\qquad\text{and}\qquad
\eta^x:=\zeta^x\ind_{\{q^x_{\zeta^x-}>0\}}+(+\infty)\ind_{\{q^x_{\zeta^x-}=0\}},
\qquad\text{for $x\in E$,}
\] 
and suppose that there exists a set $B\in\cE$ with $\lambda(B)>0$ such that $\prob(\eta^x<+\infty)>0$, for all $x\in B$.
By \cite[Corollary 1.11]{Jac85}, it holds that $\eta^L=\zeta^L=+\infty$ a.s.
For each $x\in B$, define the $\bF$-martingale $M^x=(M^x_t)_{t\in[0,T]}$ by
$
M^x := -(\ind_{\dbra{\eta^x,T}}-(\ind_{\dbra{\eta^x,T}})^p)
$,
where $(\ind_{\dbra{\eta^x,T}})^p$ denotes the dual $\bF$-predictable projection of the process $\ind_{\dbra{\eta^x,T}}$. Since $L^1(\Omega,\F_T,\prob)$ is separable, \cite[Proposition 4]{SY} ensures the existence of a $(\cP(\bF)\otimes\cE)$-measurable version of $(\ind_{\dbra{\eta^x,T}})^p$.
As a consequence of Assumption \ref{ass:Q} together with \cite[Proposition 4.9]{Fon2015}, there exists a $(\cP(\bF)\otimes\cE)$-measurable process $H^x\in L(S,\bF)$ such that $M^x=H^x\cdot S$, for every $x\in E$. The same arguments used in the proof of \cite[Proposition 4.10]{Fon2015} allow to show that the $\bG$-predictable process $H^L$ belongs to $L(S,\bG)$ and it holds that $H^L\cdot S=M^L=(\ind_{\dbra{\eta^x,T}})^p\bigr|_{x=L}$. The process $(\ind_{\dbra{\eta^x,T}})^p\bigr|_{x=L}$ is non-negative, non-decreasing and it holds that
\begin{align*}
\expec\Bigl[(\ind_{\dbra{\eta^x,T}})^p_T\bigr|_{x=L}\Bigr]
&= \int_E\expec\bigl[q^x_T(\ind_{\dbra{\eta^x,T}})_T^p\bigr]\lambda(\!\ud x)
= \int_E\expec\left[\int_0^Tq^x_{u-}\ud(\ind_{\dbra{\eta^x,T}})^p_u\right]\lambda(\!\ud x)	\\
&= \int_E\expec\left[q^x_{\eta^x-}\ind_{\{\eta^x\leq T\}}\right]\lambda(\!\ud x) > 0,
\end{align*}
where the second and third equalities follow from \cite[Theorems 5.32-5.33]{MR1219534}. This contradicts the validity of NUPBR on $\basisgp$, thus proving the first assertion of the theorem.

Under the assumption that NUPBR holds on $\basisgp$, let us prove that $\Z=\{Z/q^L\}$.
Since $S\in\Mloc(\qprob,\bF)$ and $Z$ is the density process of $\qprob$ with respect to $\prob$ on $\bF$, the processes $Z$ and $ZS$ are local martingales on $\basisp$. Hence, Lemma \ref{lem:mgG} implies that $Z/q^L\in\Mloc(\prob,\bG)$ and $ZS/q^L\in\Mloc(\prob,\bG)$, meaning that $Z/q^L\in\Z$. 
In order to prove that $\Z=\{Z/q^L\}$, let $D=(D_t)_{t\in[0,T]}$ be an arbitrary element of $\Z$ and $(\tau_n)_{n\in\N}$ a sequence of $\bG$-stopping times increasing a.s. to infinity such that $(Z/q^L)^{\tau_n}\in\M(\prob,\bG)$ and $D^{\tau_n}\in\M(\prob,\bG)$, for all $n\in\N$. For each $n\in\N$, define the filtration $\bG^n=(\mathcal G_{t\wedge\tau_n})_{t\in[0,T]}$ and the probability measures $\qprob^{\bG,n}$ and $\widehat{\qprob}^n$ by $\ud\qprob^{\bG,n}=Z_{T\wedge\tau_n}/q^L_{T\wedge\tau_n}\ud\prob$ and $\ud\widehat{\qprob}^n=D_{T\wedge\tau_n}\ud\prob$, respectively. 
It holds that $S^{\tau_n}\in\Mloc(\qprob^{\bG,n},\bG^n)\cap\Mloc(\widehat{\qprob}^n,\bG^n)$. 
Let $N=(N_t)_{t\in[0,T]}\in\M(\qprob^{\bG,n},\bG^n)$, so that $N(Z/q^L)^{\tau_n}\in\M(\prob,\bG^n)$. By Proposition \ref{prop:MRT} and \cite[Lemma 13.8]{MR1219534}, it holds that 
$N = N_0 + \gamma\cdot S^{\tau_n}$,
for some $\gamma\in L(S^{\tau_n},\bG^n)$.
It follows that $S^{\tau_n}$ has the martingale representation property on $(\Omega,\bG^n,\qprob^{\bG,n})$. In view of \cite[Corollary 11.4]{MR542115} (extended to a non-trivial initial sigma-field), this implies that $\qprob^{\bG,n}=\widehat{\qprob}^n$, for every $n\in\N$, equivalently $Z_{t\wedge\tau_n}/q^L_{t\wedge\tau_n}=D_{t\wedge\tau_n}$ a.s., for every $t\in[0,T]$ and $n\in\N$.
Since the stopping times $(\tau_n)_{n\in\N}$ increase a.s. to infinity, we obtain $Z/q^L=D$ up to an evanescent set, thus proving that $\Z=\{Z/q^L\}$.

Let us now prove the second part of the theorem:\\
(ii) $\Leftrightarrow$ (iii): This implication is evident (compare with \cite[Lemma 2.2]{Amen}).\\
(iii) $\Leftrightarrow$ (iv): It suffices to note that, as a consequence of Lemma \ref{lem:prelim_G} and formula \eqref{expec_init},
\[
\expec\left[\frac{1}{q^L_T}\right]
= \expec\left[\frac{1}{q^L_T}\ind_{\{q^L_T>0\}}\right]
= \int_E\prob(q^x_T>0)\lambda(\!\ud x).
\]
(iii) $\Leftrightarrow$ (v): Similarly as above, it holds that
\[
\expec\left[\frac{Z_T}{q^L_T}\right]
= \expec\left[\frac{Z_T}{q^L_T}\ind_{\{q^L_T>0\}}\right] 
= \int_E\expec[Z_T\ind_{\{q^x_T>0\}}]\lambda(\!\ud x)
= \int_E\qprob(q^x_T>0)\lambda(\!\ud x).
\]
Since $\qprob\sim\prob$, property (iii) is equivalent to $\qprob(q^x_T>0)=1$ for $\lambda$-a.e. $x\in E$, thus proving the claim. 	\\
(iv) $\Leftrightarrow$ (vi): Lemma \ref{lem:mgG} implies that $1/q^L\in\Mloc(\prob,\bG)$. Being a strictly positive local martingale, and therefore a supermartingale, $1/q^L$ is a true martingale if and only if $\expec[1/q^L_T]=1$.\\
(iii) $\Rightarrow$ (vii): Let define the set $B:=\{x\in E : \prob(q^x_T=0)>0\}$, with $\lambda(B)=0$. Let $N\in\M(\prob,\bF)$ and take arbitrary $0\leq s\leq t\leq T$, an $\F_s$-measurable set $A_s$ and a bounded $\cE$-measurable function $h:E\rightarrow\R$. 
By relying on formula \eqref{expec_init}, we can compute:
\[	\ba
\expec\left[\frac{N_t}{q^L_t}h(L)\ind_{A_s}\right]
&= \expec\left[\frac{N_t}{q^L_t}h(L)\ind_{A_s}\ind_{\{q^L_t>0\}}\right]
= \int_Eh(x)\expec[N_t\ind_{A_s}\ind_{\{q^x_t>0\}}]\lambda(\!\ud x)	\\
&=  \int_{E\setminus B}h(x)\expec[N_t\ind_{A_s}]\lambda(\!\ud x)
= \int_{E\setminus B}h(x)\expec[N_s\ind_{A_s}]\lambda(\!\ud x)	\\
&= \int_Eh(x)\expec[N_s\ind_{A_s}\ind_{\{q^x_s>0\}}]\lambda(\!\ud x)
= \expec\left[\frac{N_s}{q^L_s}h(L)\ind_{A_s}\right].
\ea	\]
Since the the sigma-field $\G_s$ is generated by random variables of the form $h(L)\ind_{A_s}$, this proves that $N/q^L\in\M(\prob,\bG)$.\\
(vii) $\Rightarrow$ (i): We already know that $Z/q^L\in\Z$. Since $Z\in\M(\prob,\bF)$, property (vii) implies that $Z/q^L\in\M(\prob,\bG)$. Hence, the probability measure $\qprob^{\bG}$ defined by $\ud\qprob^{\bG}:=Z_T/q^L_T\ud\prob$ is an equivalent local martingale measure for $S$ on $\basisgp$. By \cite{DS98b}, $S$ satisfies NFLVR on $\basisgp$.\\
(i) $\Rightarrow$ (v): We argue by contradiction and construct  an arbitrage opportunity explicitly. Suppose that $\expec[Z_T/q^L_T]\neq 1$. Since $Z/q^L$ is a supermartingale on $\basisgp$, it must be that $\expec[Z_T/q^L_T]<1$. Define $M=(M_t)_{t\in[0,T]}\in\M(\prob,\bG)$ by $M_t:=\expec[Z_T/q^L_T|\G_t]$, for all $t\in[0,T]$. By Proposition \ref{prop:MRT}, there exists $K\in L(S,\bG)$ such that $M_t=Z_t/q^L_t(M_0+(K\cdot S)_t)$ a.s. for all $t\in[0,T]$. Note that
\[
(K\cdot S)_t = \frac{q^L_t}{Z_t}M_t - M_0
\geq -M_0 \geq -1
\qquad\text{ a.s. for all $t\in[0,T]$,}
\]
where the last inequality follows from the $\bG$-supermartingale property of $Z/q^L$. Therefore, the strategy $K$ is $1$-admissible, in the sense of \cite{DS94}. Moreover, it holds that $(K\cdot S)_T = 1-M_0 \geq 0$ a.s. and $\expec[M_0]<1$ implies that $\prob((K\cdot S)_T>0)>0$. This shows that $K$ is an arbitrage opportunity, thus contradicting the validity of NFLVR in $\basisgp$.
\end{proof}

\subsection{Proofs of the results of Section \ref{sec:portfolio}}

We first prove Lemmata \ref{lem:AFinAG} and \ref{lem:SuperMart=Mart}, which together provide a complete duality description of the set of admissible portfolios.

\begin{proof}[Proof of Lemma \ref{lem:AFinAG}]
(i): 
Let $\mathbf{H}\in\{\bF,\bG\}$ and $(\vartheta,c)\in\A_+^{\mathbf{H},k}(v)$.
For simplicity of notation, let us denote $V:=V^{v+k,\vartheta,c}$, $C:=\int_0^{\cdot}c_u\ud\kappa_u$ and $\widetilde{C}:=\int_0^{\cdot}Z^{\mathbf{H}}_u\ud C_u$.  
By the integration by parts formula (see, e.g., \cite[Proposition I.4.49]{MR1943877}), we have that, for all $t\in[0,T]$,
\[
Z^{\mathbf{H}}_tV_t + \widetilde{C}_t
= Z^{\mathbf{H}}_t\bigl(v+k+(\vartheta\cdot S)_t\bigr) - Z^{\mathbf{H}}_tC_t + \int_0^tZ^{\mathbf{H}}_u\ud C_u
= Z^{\mathbf{H}}_t\bigl(v+k+(\vartheta\cdot S)_t\bigr) - (C_-\cdot Z^{\mathbf{H}})_t.
\]
Since $Z^{\mathbf{H}}\in\Mloc(\prob,\mathbf{H})$ and $Z^{\mathbf{H}}S\in\Mloc(\prob,\mathbf{H})$, this implies that $Z^{\mathbf{H}}V + \widetilde{C}$ is a sigma-martingale on $(\Omega,\mathbf{H},\prob)$ (see, e.g., \cite[Lemma 4.2]{Fon15}). Being non-negative, it is also  a supermartingale (see \cite[Proposition 3.1]{MR2013413}).
In order to show that $(\vartheta,c)\in\A_{sm}^{\mathbf{H},k}(v)$, it remains to prove that $\widetilde{C}_T\in L^1(\prob)$. Arguing similarly as in \cite[Lemma 1]{CCFM2015}, let $(\tau_n)_{n\in\N}$ be a localizing sequence of $\mathbf{H}$-stopping times for $C_-\cdot Z^{\mathbf{H}}\in\Mloc(\prob,\mathbf{H})$. Then, for every $n\in\N$, it holds that
\[
v+k 
\geq \expec\left[Z^{\mathbf{H}}_{T\wedge\tau_n}\bigl(v+k+(\vartheta\cdot S)_{T\wedge\tau_n}\bigr)\right]
\geq \expec\bigl[Z^{\mathbf{H}}_{T\wedge\tau_n}C_{T\wedge\tau_n}\bigr]
= \expec\left[\int_0^{T\wedge\tau_n}C_{u-}\ud Z^{\mathbf{H}}_u + \int_0^{T\wedge\tau_n}Z^{\mathbf{H}}_u\ud C_u\right],
\]
where we have used the supermartingale property of $Z^{\mathbf{H}}(v+k+\vartheta\cdot S)$ and integration by parts.
Since $(C_-\cdot Z^{\mathbf{H}})^{\tau_n}\in\M(\prob,\mathbf{H})$, for all $n\in\N$, the monotone convergence theorem yields that
\[
v+k 
\geq \lim_{n\rightarrow+\infty} \expec\bigl[\widetilde{C}_{T\wedge\tau_n}\bigr]
= \expec\bigl[\widetilde{C}_T\bigr].
\]
Conversely, let $(\vartheta,c)\in\A_{sm}^{\mathbf{H},k}(v)$. By definition, the process $Z^{\mathbf{H}}V^{v+k,\vartheta,c}+\int_0^{\cdot}Z^{\mathbf{H}}_uc_u\ud\kappa_u$ is a supermartingale on $(\Omega,\mathbf{H},\prob)$. 
Therefore, for all $t\in[0,T]$,
\[
Z^{\mathbf{H}}_tV^{v+k,\vartheta,c}_t + \int_0^tZ^{\mathbf{H}}_uc_u\ud\kappa_u
\geq \expec\left[Z^{\mathbf{H}}_TV^{v+k,\vartheta,c}_T+\int_0^TZ^{\mathbf{H}}_uc_u\ud\kappa_u\bigg|\mathcal{H}_t\right]
\geq \expec\left[\int_0^TZ^{\mathbf{H}}_uc_u\ud\kappa_u\bigg|\mathcal{H}_t\right],
\]
so that
$
Z^{\mathbf{H}}_tV^{v+k,\vartheta,c}_t 
\geq \expec[\int_t^TZ^{\mathbf{H}}_uc_u\ud\kappa_u|\mathcal{H}_t]\geq0.
$
This shows that $V^{v,\vartheta,c}_t\geq-k$ a.s. for all $t\in[0,T]$, thus proving that $(\vartheta,c)\in\A_+^{\mathbf{H},k}(v)$.\\
(ii): 
Lemma \ref{lem:prelim_G} together with \cite[Proposition 2.1]{MR604176} implies that $L(S,\bF)\subseteq L(S,\bG)$, from which the inclusion $\A^{\bF,k}_+(v)\subseteq\A^{\bG,k}_+(v)$ immediately follows. In turn, in view of part (i) of the lemma, this  implies that $\A^{\bF,k}_{sm}(v)\subseteq\A^{\bG,k}_{sm}(v)$.

Let us now turn to the proof of the last assertion of the lemma.
Suppose first that $\expec[1/q^L_T]\neq 1$, or, equivalently, $\expec[Z_T/q^L_T]\neq 1$ (see Theorem \ref{thm:results_NUPBR}). Since $Z/q^L$ is a strictly positive local martingale on $\basisgp$ (see Lemma \ref{lem:mgG}), it must be that $\expec[Z_T/q^L_T]<1$ and $Z/q^L\notin\M(\prob,\bG)$ (i.e., $Z/q^L$ is a {\em strict} local martingale).
Consider the pair $(0,0)\in\A^{\bF,k}_m(1)$, generating the constant value process $V^{1,0,0}=1$. Since $Z^{\bG}V^{1,0,0}=Z/q^L\notin\M(\prob,\bG)$, this suffices to show that $(0,0)\notin\A_m^{\bG,k}(1)$, thus proving that $\A^{\bF,k}_m(1)\nsubseteq\A^{\bG,k}_m(1)$. 
Conversely, suppose that $\expec[1/q^L_T]=1$ and let $(\vartheta,c)\in\A^{\bF,k}_m(v)$. Let us denote $V:=V^{v+k,\vartheta,c}$ and $\widetilde{C}:=\int_0^{\cdot}Z^{\bF}_uc_u\ud\kappa_u$ and note that
\[
\expec\left[\frac{\widetilde{C}_T}{q^L_T}\right]
= \expec\left[\frac{\widetilde{C}_T}{q^L_T}\ind_{\{q^L_T>0\}}\right]
= \int_E\expec\bigl[\widetilde{C}_T\ind_{\{q^x_T>0\}}\bigr]\lambda(\!\ud x)
= \expec\bigl[\widetilde{C}_T\bigr] < +\infty,
\]
where we have used formula \eqref{expec_init} and the equivalence (iii)$\Leftrightarrow$(iv) of Theorem \ref{thm:results_NUPBR}.
Again by Theorem \ref{thm:results_NUPBR}, it holds that $1/q^L\in\M(\prob,\bG)$. Therefore, taking the $\bG$-optional projection of $1/q^L_T$ (see, e.g., \cite[Theorem 5.16]{MR1219534}), we can write, for all $t\in[0,T]$,
\begin{align*}
\expec\left[\frac{\widetilde{C}_T}{q^L_T}\biggr|\G_t\right]
&= \expec\left[\int_t^T\frac{1}{q^L_T}\ud\widetilde{C}_u+\frac{\widetilde{C}_t}{q^L_T}\biggr|\G_t\right] 
= \expec\left[\int_t^T\frac{1}{q^L_u}\ud\widetilde{C}_u\biggr|\G_t\right] +\frac{\widetilde{C}_t}{q^L_t} \\
&= \expec\left[\int_0^T\frac{1}{q^L_u}\ud\widetilde{C}_u\biggr|\G_t\right] - \int_0^t\frac{1}{q^L_u}\ud\widetilde{C}_u + \frac{\widetilde{C}_t}{q^L_t},
\end{align*}
thus showing that $\widetilde{C}/q^L-\int_0^{\cdot}(1/q^L_u)\ud\widetilde{C}_u\in\M(\prob,\bG)$.
Since $(\vartheta,c)\in\A^{\bF,k}_m(v)$, it holds that $Z^{\bF}V+\widetilde{C}\in\M(\prob,\bF)$. By Theorem \ref{thm:results_NUPBR}, the fact that $\expec[1/q^L_T]=1$ then implies that $(Z^{\bF}V+\widetilde{C})/q^L\in\M(\prob,\bG)$. By the integration by parts formula, we have that
\[
Z^{\bG}_tV_t + \int_0^tZ^{\bG}_u\ud C_u
= \frac{Z^{\bF}_tV_t}{q^L_t} + \int_0^t\frac{1}{q^L_u}\ud\widetilde{C}_u
= \frac{Z^{\bF}_tV_t+\widetilde{C}_t}{q^L_t} - \frac{\widetilde{C}_t}{q^L_t} + \int_0^t\frac{1}{q^L_u}\ud\widetilde{C}_u,
\]
for all $t\in[0,T]$.
This proves that $Z^{\bG}V+\int_0^{\cdot}Z^{\bG}_uc_u\ud\kappa_u\in\M(\prob,\bG)$, so that $(\vartheta,c)\in\A^{\bG,k}_m(v)$.
\end{proof}

\begin{proof}[Proof of Lemma \ref{lem:SuperMart=Mart}]
The fact that $\Cons_+^{\mathbf{H},k}(v)=\Cons_{sm}^{\mathbf{H},k}(v)$ is a direct consequence of part (i) of Lemma \ref{lem:AFinAG}.
If $c\in\Cons_+^{\mathbf{H},k}(v)$, then there exists $\vartheta\in L(S,\mathbf{H})$ such that $(\vartheta,c)\in\A_+^{\mathbf{H},k}(v)=\A_{sm}^{\mathbf{H},k}(v)$ (see Lemma \ref{lem:AFinAG}). Therefore, due to the supermartingale property of the process $Z^{\mathbf{H}}V^{v+k,\vartheta,c}+\int_0^{\cdot}Z^{\mathbf{H}}_u\,c_u\ud\kappa_u$ and to the fact that $V^{v,\vartheta,c}_T\geq0$ a.s., it holds that
\[
v+k \geq \expec\left[Z^{\mathbf{H}}_TV^{v+k,\vartheta,c}_T+\int_0^TZ^{\mathbf{H}}_uc_u\ud\kappa_u\biggr|\mathcal{H}_0\right]
\geq  \expec\left[kZ^{\mathbf{H}}_T + \int_0^TZ^{\mathbf{H}}_uc_u\ud\kappa_u\biggr|\mathcal{H}_0\right],
\]
so that $\expec[\int_0^TZ^{\mathbf{H}}_uc_u\ud\kappa_u|\mathcal{H}_0]\leq v+k(1-\expec[Z^{\mathbf{H}}_T|\mathcal{H}_0])$ a.s.
Conversely, let $C:=\int_0^{\cdot}c_u\ud\kappa_u$ and suppose that $\expec[\int_0^TZ^{\mathbf{H}}_u\ud C_u|\mathcal{H}_0]\leq v+k(1-\expec[Z^{\mathbf{H}}_T|\mathcal{H}_0])$ a.s. 
Consider the process $\widehat{V}=(\widehat{V}_t)_{t\in[0,T]}$ defined by
\begin{align*}
\widehat{V}_t &:= v + Z^{\mathbf{H}}_tC_t - \int_0^tZ^{\mathbf{H}}_u\ud C_u  
+ \expec\bigg[ \int_0^TZ^{\mathbf{H}}_u\ud C_u\bigg|\mathcal{H}_t\bigg] 
- \expec\bigg[ \int_0^TZ^{\mathbf{H}}_u\ud C_u\bigg|\mathcal{H}_0\bigg]	\\
&\qquad +k\left(1-\expec[Z^{\mathbf{H}}_T|\mathcal{H}_0]+\expec[Z^{\mathbf{H}}_T|\mathcal{H}_t]-Z^{\mathbf{H}}_t\right),
\end{align*}
for all $t\in[0,T]$. The process $\widehat{V}$ is well-defined as an element of $\Mloc(\prob,\mathbf{H})$. As a consequence of Assumption \ref{ass:Q} (and of Proposition \ref{prop:MRT}, in the case $\mathbf{H}=\bG$), there exists $\psi\in L(S,\mathbf{H})$ such that
\[
\widehat{V}_t = Z^{\mathbf{H}}_t\bigl(v + (\psi\cdot S)_t\bigr)
\qquad\text{ a.s. for all }t\in[0,T].
\]
The process $V^{v+k,\psi,c}=(V^{v+k,\psi,c}_t)_{t\in[0,T]}$ associated to the pair $(\psi,c)$ satisfies
\begin{align*}
Z^{\mathbf{H}}_tV^{v+k,\psi,c}_t + \int_0^tZ^{\mathbf{H}}_u\ud C_u 
&= v+k +  \expec\bigg[\int_0^TZ^{\mathbf{H}}_u\ud C_u\bigg|\mathcal{H}_t\bigg] - \expec\bigg[\int_0^TZ^{\mathbf{H}}_u\ud C_u\bigg|\mathcal{H}_0\bigg]	\\
&\quad+k\left(\expec[Z^{\mathbf{H}}_T|\mathcal{H}_t]-\expec[Z^{\mathbf{H}}_T|\mathcal{H}_0]\right)
\qquad\qquad\text{ a.s. for all }t\in[0,T].
\end{align*}
By construction, it holds that $Z^{\mathbf{H}}_TV^{v,\psi,c}_T\geq0$ a.s.
This shows that $(\psi,c)\in\A^{\mathbf{H},k}_{m}(v)\subseteq\A_+^{\mathbf{H},k}(v)$ (see Lemma \ref{lem:AFinAG}), thus proving that $c\in\Cons_+^{\mathbf{H},k}(v)$.
Assertion (ii) follows by analogous arguments, using the definition of the set $\Cons_m^{\mathbf{H},k}(v)$ and replacing the supermartingale property with the martingale property.
It remains to prove \eqref{eq:utility_mart}.
Since $\A_m^{\mathbf{H},k}(v)\subseteq\A_+^{\mathbf{H},k}(v)$, it is clear that $u^{\mathbf{H},k}(v)\geq u_m^{\mathbf{H},k}(v)$. In order to prove the converse inequality, let $c\in\Cons_+^{\mathbf{H},k}(v)$. By part (i) of the lemma, it holds that $\expec[\int_0^TZ^{\mathbf{H}}_u\,c_u\ud\kappa_u|\mathcal{H}_0]\leq v+k(1-\expec[Z^{\mathbf{H}}_T|\mathcal{H}_0])$ a.s.
Let define the $\mathcal{H}_0$-measurable non-negative random variable $\tilde{v}:=v+k(1-\expec[Z^{\mathbf{H}}_T|\mathcal{H}_0])-\expec[\int_0^TZ^{\mathbf{H}}_u\,c_u\ud\kappa_u|\mathcal{H}_0]$ and define $\tilde{c}=(\tilde{c}_t)_{t\in[0,T]}\in\cO_+(\mathbf{H})$ by
\[
\tilde{c}_t := c_t + \frac{\tilde{v}}{Z^{\mathbf{H}}_t\,\expec[\kappa_T|\mathcal{H}_0]},
\qquad\text{ for all }t\in[0,T].
\]
By construction, it holds that $\expec[\int_0^TZ^{\mathbf{H}}_u\tilde{c}_u\ud\kappa_u|\mathcal{H}_0]=v+k(1-\expec[Z^{\mathbf{H}}_T|\mathcal{H}_0])$ a.s., so that $\tilde{c}\in\Cons_m^{\mathbf{H},k}(v)$.
Moreover, for every $t\in[0,T]$, we have that $\prob(\tilde{c}_t\geq c_t)=1$, with $\prob(\tilde{c}_t>c_t)>0$ if and only if $c\notin\Cons_m^{\mathbf{H},k}(v)$. Since $U$ is assumed to be strictly increasing (Assumption \ref{ass:U}), this implies that
\[
\expec\left[\int_0^TU(u,c_u)\ud\kappa_u\right]
\leq \expec\left[\int_0^TU(u,\tilde{c}_u)\ud\kappa_u\right],
\]
with strict inequality holding if and only if $c\notin\Cons_m^{\mathbf{H}}(v)$. 
By the arbitrariness of $c\in\Cons_+^{\mathbf{H},k}(v)$, we then have $u^{\mathbf{H},k}(v)\leq u_m^{\mathbf{H},k}(v)$, thus proving equality \eqref{eq:utility_mart}. 
\end{proof}

We are now in a position to prove Proposition \ref{prop:lambda}.
Even though the proof follows a well-known scheme, we give full details for the convenience of the reader.

\begin{proof}[Proof of Proposition \ref{prop:lambda}]
Under the present assumptions, the process $c^{\mathbf{H}}=(c^{\mathbf{H}}_t)_{t\in[0,T]}$ satisfies $\expec[\int_0^TZ^{\mathbf{H}}_uc^{\mathbf{H}}_u\ud\kappa_u|\mathcal{H}_0]=v+k(1-\expec[Z^{\mathbf{H}}|\mathcal{H}_0])$ a.s., so that $c^{\mathbf{H}}\in\Cons^{\mathbf{H},k}_m(v)$ by Lemma \ref{lem:SuperMart=Mart}.
Consider an arbitrary consumption process $c\in\Cons_m^{\mathbf{H},k}(v)$. By the concavity of $U$ (Assumption \ref{ass:U}), it holds that
\[
U(t,c^{\mathbf{H}}_t) 
\geq U(t,c_t) + U'(t,c^{\mathbf{H}}_t)(c^{\mathbf{H}}_t - c_t)  
= U(t,c_t) + \Lambda^{\mathbf{H}}(v)Z^{\mathbf{H}}_t(c^{\mathbf{H}}_t - c_t),
\quad\text{ for every }t\in[0,T].
\]
Therefore,
\begin{align*}
\expec\left[  \int_0^TU(u,c^{\mathbf{H}}_u)\ud\kappa_u 
\bigg|\mathcal{H}_0\right] 
&\geq \expec\left[  \int_0^TU(u,c_u)\ud\kappa_u \bigg|\mathcal{H}_0\right]  \\
&\quad + \Lambda^{\mathbf{H}}(v) \expec \left[ \int_0^T Z^{\mathbf{H}}_uc^{\mathbf{H}}_u\ud\kappa_u    \bigg|\mathcal{H}_0\right]
-  \Lambda^{\mathbf{H}}(v) \expec \left[ \int_0^TZ^{\mathbf{H}}_uc_u\ud\kappa_u  \bigg|\mathcal{H}_0\right]\\
&= \expec\left[  \int_0^TU(u,c_u)\ud\kappa_u \bigg|\mathcal{H}_0\right],
\end{align*}
where the equality follows from the fact that, in view of part (ii) of Lemma \ref{lem:SuperMart=Mart}, 
\[
v + k\left(1-\expec[Z^{\mathbf{H}}_T|\mathcal{H}_0]\right)
= \expec \left[ \int_0^TZ^{\mathbf{H}}_uc^{\mathbf{H}}_u\ud\kappa_u  \bigg|\mathcal{H}_0\right]
= \expec \left[ \int_0^TZ^{\mathbf{H}}_uc_u\ud\kappa_u  \bigg|\mathcal{H}_0\right]
\text{ a.s.}
\]
The claim follows by the arbitrariness of $c\in\Cons_m^{\mathbf{H},k}(v)$ together with equality \eqref{eq:utility_mart} in Lemma \ref{lem:SuperMart=Mart}.
\end{proof}

\begin{proof}[Proof of Corollary \ref{cor:log_pwr_utility}]
In view of Proposition \ref{prop:lambda}, in order to compute $u^{\mathbf{H},k}(v)$ it suffices to find explicitly the $\mathcal{H}_0$-measurable random variable $\Lambda^{\mathbf{H},k}(v)$ satisfying equation \eqref{eq:condition_lambda}. 
Note that, whenever it exists, the random variable $\Lambda^{\mathbf{H},k}(v)$ is uniquely determined (up to a $\prob$-nullset).

(i): If $U(\omega,t,x)=\log(x)$, then $I(\omega,t,y) = 1/y$, for all $(\omega,t,y)\in\Omega\times[0,T]\times(0,+\infty)$. Therefore, equation \eqref{eq:condition_lambda} can be explicitly solved and it holds that $\Lambda^{\mathbf{H},k}(v) = \expec[\kappa_T|\mathcal{H}_0]/(v+k(1-\expec[Z^{\mathbf{H}}_T|\mathcal{H}_0]))$. By Proposition \ref{prop:lambda}, the optimal solution $c^{\mathbf{H}}=(c^{\mathbf{H}}_t)_{t\in[0,T]}$ is then given by 
\[
c^{\mathbf{H}}_t
= \frac{1}{\Lambda^{\mathbf{H},k}(v)Z^{\mathbf{H}}_t}
= \frac{v+k(1-\expec[Z^{\mathbf{H}}_T|\mathcal{H}_0])}{Z^{\mathbf{H}}_t\,\expec[\kappa_T|\mathcal{H}_0]}, 
\qquad\text{ for all }t\in[0,T].
\]
Under the integrability assumption stated in the corollary, the optimal expected utility $u^{\mathbf{H},k}(v)$ as given by \eqref{eq:opt_ut_log} can be obtained by means of a straightforward computation.
Note that the first two terms on the right-hand side of \eqref{eq:opt_ut_log} are always finite, as a consequence of the boundedness of $\kappa_T$.

(ii): If $U(\omega,t,x)=x^p/p$, then $I(\omega,t,y) = y^{1/(p-1)}$, for all $(\omega,t,y)\in\Omega\times[0,T]\times(0,+\infty)$. 
By Proposition \ref{prop:lambda}, the $\mathcal{H}_0$-measurable random variable $\Lambda^{\mathbf{H},k}(v)$ must solve
\[
\expec\bigg[ \int_0^T(Z^{\mathbf{H}}_u)^{\frac{p}{p-1}} \big( \Lambda^{\mathbf{H},k}(v) \big)^{\frac{1}{p-1}} \ud\kappa_u \bigg| \mathcal{H}_0 \bigg] = v+k\left(1-\expec[Z^{\mathbf{H}}_T|\mathcal{H}_0]\right).
\]
Therefore, if $\expec[\int_0^T(Z^{\mathbf{H}}_u)^{p/(p-1)}\ud\kappa_u|\mathcal{H}_0]<+\infty$ a.s., then we have that
\[
\Lambda^{\mathbf{H},k}(v) = \left(v+k(1-\expec[Z^{\mathbf{H}}_T|\mathcal{H}_0])\right)^{p-1} \expec\left[ \int_0^T (Z^{\mathbf{H}}_u)^{\frac{p}{p-1}}\ud\kappa_u \bigg| \mathcal{H}_0 \right]^{1-p}.
\]
By Proposition \ref{prop:lambda}, the corresponding optimal consumption process $c^{\mathbf{H}}=(c^{\mathbf{H}}_t)_{t\in[0,T]}$ is given by
$
c_t^{\mathbf{H}}  
= (\Lambda^{\mathbf{H},k}(v) Z^{\mathbf{H}}_t)^{1/(p-1)} 
$,
for all $t\in[0,T]$.
If $\expec[\int_0^T(Z^{\mathbf{H}}_u)^{p/(p-1)}\ud\kappa_u|\mathcal{H}_0]^{1-p}  \in L^1(\prob)$, then the optimal expected utility $u^{\mathbf{H},k}(v)$ is finite and can be explicitly computed as in \eqref{eq:opt_ut_pwr}.
\end{proof}

\begin{proof}[Proof of Corollary \ref{cor:exp_utility}]
We first show that equation \eqref{eq:exp_utility_lambda} admits an a.s. unique solution, for every $v>v^{\mathbf{H}}_k$. To this effect, we define the $\mathcal{H}_0$-measurable function $g:\Omega\times(0,+\infty)\rightarrow\R_+$ by
\[
g(\lambda) := \frac{1}{\alpha}\expec\left[\int_0^TZ^{\mathbf{H}}_u\left(\log\left(\frac{\alpha}{\lambda Z^{\mathbf{H}}_u}\right)\right)^+\!\ud\kappa_u\bigg|\mathcal{H}_0\right],
\qquad\text{ for }\lambda\in(0,+\infty).
\]
Note that $g$ is well-defined, since 
\[
g(\lambda) = \frac{1}{\alpha}\expec\left[\int_0^TZ^{\mathbf{H}}_u\log\left(\frac{\alpha}{\lambda Z^{\mathbf{H}}_u}\right)\ind_{\{Z^{\mathbf{H}}_u\leq \alpha/\lambda\}}\ud\kappa_u\bigg|\mathcal{H}_0\right]
\leq \frac{\expec[\kappa_T|\mathcal{H}_0]}{\lambda} < +\infty
\text{ a.s.}
\]
Clearly, $g$ is a decreasing function. Furthermore, the dominated convergence theorem implies that $g$ is continuous. Again by  dominated convergence, it holds that $\lim_{\lambda\rightarrow+\infty}g(\lambda)=0$ a.s. and a straightforward application of Fatou's lemma yields that $\lim_{\lambda\downarrow0}g(\lambda)=+\infty$ a.s.
Moreover, for all $0<\lambda'<\lambda<+\infty$, it holds that $g(\lambda')>g(\lambda)$ a.s. on $\{g(\lambda)>0\}$.
Indeed, arguing by contradiction, if the $\mathcal{H}_0$-measurable set $G_{\lambda,\lambda'}:=\{g(\lambda)=g(\lambda'),g(\lambda)>0\}$ has strictly positive probability, then 
\[
\expec\left[\int_0^TZ^{\mathbf{H}}_u\left(\left(\log\left(\frac{\alpha}{\lambda' Z^{\mathbf{H}}_u}\right)\right)^+-\left(\log\left(\frac{\alpha}{\lambda Z^{\mathbf{H}}_u}\right)\right)^+\right)\!\ud\kappa_u\bigg|\mathcal{H}_0\right] = 0
\qquad\text{ on }G_{\lambda,\lambda'}.
\]
However, since $\log(\alpha/(\lambda'Z^{\mathbf{H}}_u))>\log(\alpha/(\lambda Z^{\mathbf{H}}_u))$ for all $u\in[0,T]$, this contradicts the assumption that $g(\lambda)>0$.
In view of these observations, $v+k(1-\expec[Z^{\mathbf{H}}_T|\mathcal{G}_0](\omega))\in\{g(\lambda)(\omega):\lambda\in(0,+\infty)\}$ for a.a. $\omega\in\Omega$.
Therefore, by \cite[Lemma 1]{Benes70}, equation \eqref{eq:exp_utility_lambda} admits a unique strictly positive $\mathcal{H}_0$-measurable solution $\Lambda^{\mathbf{H},k}(v)$, for every $v>v^{\mathbf{H}}_k$.
With some abuse of notation, let $U(x):=-e^{-\alpha x}$, for $x\in\R_+$, and $I(y):=(1/\alpha)(\log(\alpha/y))^+$, for $y\in(0,+\infty)$. It can be easily checked that
\be	\label{eq:duality_exp}
\sup_{x\in\R_+}\bigl(U(x)-xy\bigr)
= U\bigl(I(y)\bigr)-yI(y),
\qquad\text{ for all }y>0.
\ee
Let us then define the consumption process $c^{\mathbf{H}}=(c^{\mathbf{H}}_t)_{t\in[0,T]}$ by
\[
c^{\mathbf{H}}_t := \frac{1}{\alpha}\left(\log\left(\frac{\alpha}{\Lambda^{\mathbf{H},k}(v)Z^{\mathbf{H}}_t}\right)\right)^+,
\qquad\text{ for all }t\in[0,T],
\]
and note that $\expec[\int_0^TZ^{\mathbf{H}}_uc^{\mathbf{H}}_u\ud\kappa_u|\mathcal{H}_0]=v+k(1-\expec[Z^{\mathbf{H}}_T|\mathcal{H}_0])$ a.s., so that $c^{\mathbf{H}}\in\Cons_m^{\mathbf{H},k}(v)$.
Consider an arbitrary consumption process $c=(c_t)_{t\in[0,T]}\in\Cons_m^{\mathbf{H},k}(v)$ and note that, as a consequence of \eqref{eq:duality_exp}, 
\[
U(c^{\mathbf{H}}_t)
\geq U(c_t) + \Lambda^{\mathbf{H}}(v)Z^{\mathbf{H}}_t(c^{\mathbf{H}}_t-c_t),
\qquad\text{ for all }t\in[0,T].
\]
The same argument adopted in the proof of Proposition \ref{prop:lambda} allows then to conclude that $c^{\mathbf{H}}$ is the optimal consumption process.
Formula \eqref{eq:opt_exp_utility} then follows by direct computations.
\end{proof}

\subsection{Proofs of the results of Section \ref{sec:indiff_price}}
\label{sec:proofs_3}

\begin{proof}[Proof of Theorem \ref{thm:IndPrice}]
(i): Due to the concavity of $U$ (see Assumption \ref{ass:U}), the assumption that $u^{\bG,k}(v_0)<+\infty$ for some $v_0>v^{\bG}_k$ implies that the function $u^{\bG,k}$ is concave and $u^{\bG,k}(v)<+\infty$, for all $v\geq v^{\bG}_k$.
By Lemma \ref{lem:AFinAG}, it holds that $u^{\bG,k}(v)\geq u^{\bF,k}(v)=u^{\bF,0}(v)>-\infty$, for every $v>0$.
Therefore, for every $v>0$, equation \eqref{eq:price_arb} admits a unique non-negative solution $\pi^{U,k}(v)$ if the function $u^{\bG,k}$ is  continuous, strictly increasing and satisfies $\lim_{w\searrow v^{\bG}_k}u^{\bG,k}(w)<u^{\bF,0}(v)$.
Under the present assumptions, $u^{\bG,k}$ satisfies these properties.
Indeed, by concavity, the function $u^{\bG,k}$ is continuous on $(v^{\bG}_k,+\infty)$.
Moreover, by condition \eqref{eq:condition_lambda}, we have that, for every $v>v^{\bG}_k$ and $\delta>0$,
\[
\expec\biggl[\int_0^TZ^{\bG}_u\Bigl(I\bigl(u,\Lambda^{\bG,k}(v+\delta)Z^{\bG}_u\bigr)-I\bigl(u,\Lambda^{\bG,k}(v)Z^{\bG}_u\bigr)\Bigr)\!\ud\kappa_u\,\biggr|\G_0\biggr]
= \delta,
\]
for some $\G_0$-measurable random variables $\Lambda^{\bG,k}(v+\delta)$ and $\Lambda^{\bG,k}(v)$.
Since $Z^{\bG}>0$ and $I(\omega,t,\cdot)$ is strictly decreasing, for every $(\omega,t)\in\Omega\times[0,T]$, this implies that $\Lambda^{\bG,k}(v+\delta)<\Lambda^{\bG,k}(v)$ a.s.
Recalling that $u^{\bG,k}(v)=\expec[\int_0^TU(u,I(u,\Lambda^{\bG,k}(v)Z^{\bG}_u))\!\ud\kappa_u]$, as a consequence of Proposition \ref{prop:lambda}, this implies that $u^{\bG,k}$ is strictly increasing.

(ii): If $\expec[1/q^L_T]<1$, then $\expec[Z^{\bG}_T]<1$ (see Theorem \ref{thm:results_NUPBR}). As explained in Remark \ref{rem:credit_cons} (with $\mathbf{H}=\bG$), this entails that $k\mapsto u^{\bG,k}(v)$ is strictly increasing. In turn, in view of Definition \ref{def:uip}, this implies that $k\mapsto\pi^{U,k}(v)$ is strictly increasing, for every $v>0$.
Conversely, if the map $k\mapsto\pi^{U,k}(v)$ is strictly increasing, then it necessarily holds that $u^{\bG,k}(v)>u^{\bG,0}(v)$, for every $v>v^{\bG}_k$. In view of Remark \ref{rem:credit_arbitrage} together with Theorem \ref{thm:results_NUPBR}, this implies that $\expec[1/q^L_T]<1$.

(iii): It suffices to show that, if $\int_E(\expec[\int_0^T\!\ind_{\{q^x_t=0\}}\!\ud\kappa_t]+k\prob(q^x_T=0))\lambda(\!\ud x)>0$ holds, then $u^{\bG,k}(v)>u^{\bF,0}(v)$, for all $k\in\R_+$ and $v>0$. 
Under the present assumptions and in view of Lemma \ref{lem:SuperMart=Mart}, there exists a pair $(\vartheta^{\bF},c^{\bF})\in\A^{\bF,0}_m(v)$ such that $c^{\bF}$ solves problem \eqref{eq:uH} (with $\mathbf{H}=\bF$).
By Lemma \ref{lem:AFinAG}, it holds that $(\vartheta^{\bF},c^{\bF})\in\A^{\bF,k}_m(v)\subseteq\A^{\bG,k}_{sm}(v)$, so that
\[
M_0 := \expec\left[\int_0^TZ^{\bG}_uc^{\bF}_u\ud\kappa_u+kZ^{\bG}_T\bigg|\G_0\right]  \leq v+k
\quad\text{a.s.}
\] 
By formula \eqref{expec_init}, the $\G_0$-measurable random variable $M_0$ can be computed explicitly. Indeed, let $h:E\rightarrow\R$ be an arbitrary $\cE$-measurable bounded function. Then
\begin{align*}
\expec\left[h(L)\left(\int_0^TZ^{\bG}_uc^{\bF}_u\ud\kappa_u+kZ^{\bG}_T\right)\right]
&= \int_Eh(x)\expec\left[q^x_T\int_0^T\frac{Z^{\bF}_u}{q^x_u}\ind_{\{q^x_u>0\}}c^{\bF}_u\ud\kappa_u+kZ^{\bF}_T\ind_{\{q^x_T>0\}}\right]\lambda(\!\ud x)	\\
&= \int_Eh(x)\expec\left[\int_0^TZ^{\bF}_u\ind_{\{q^x_u>0\}}c^{\bF}_u\ud\kappa_u+kZ^{\bF}_T\ind_{\{q^x_T>0\}}\right]\lambda(\!\ud x),
\end{align*}
where the second equality follows from \cite[Theorem 5.32]{MR1219534}. We have thus shown that $M_0=\expec[\int_0^T\!Z^{\bF}_u\ind_{\{q^x_u>0\}}c^{\bF}_u\ud\kappa_u+kZ^{\bF}_T\ind_{\{q^x_T>0\}}]\bigr|_{x=L}$ a.s.
Since the process $c^{\bF}$ is strictly positive $(\!\ud\kappa\otimes\prob)$-a.e. (as a consequence of Assumption \ref{ass:U}), the condition $\int_E(\expec[\int_0^T\!\ind_{\{q^x_t=0\}}\!\ud\kappa_t]+k\prob(q^x_T=0))\lambda(\!\ud x)>0$ implies that $\prob(M_0<v+k)>0$.
Define then an $\cO(\bG)$-measurable process $\hat{c}=(\hat{c}_t)_{t\in[0,T]}$ by
\[
\hat{c}_t := c^{\bF}_t + \frac{v+k-M_0}{Z^{\bG}_t\,\expec[\kappa_T|\G_0]},
\qquad\text{ for all }t\in[0,T].
\]
Similarly as in the last part of the proof of Lemma \ref{lem:SuperMart=Mart}, it holds that $\hat{c}\in\Cons_m^{\bG,k}(v)$. Furthermore, since $\prob(\hat{c}_t>c^{\bF}_t)>0$ for every $t\in[0,T]$, we have that
\[
u^{\bG,k}(v) \geq \expec\left[\int_0^TU(u,\hat{c}_u)\ud\kappa_u\right]
> \expec\left[\int_0^TU(u,c^{\bF}_u)\ud\kappa_u\right]
= u^{\bF,k}(v),
\]
thus completing the proof.
\end{proof}


\begin{proof}[Proof of Proposition \ref{prop:pi}]
The result follows from the explicit expressions for the optimal expected utilities obtained in Corollary \ref{cor:log_pwr_utility}. More specifically, in view of equation \eqref{eq:opt_ut_log}, part (i) follows by solving with respect to $\pi^{{\rm log}}(v)$ the equation
\begin{align*}
0 &= u^{\bF}(v)-u^{\bG}\bigl(v-\pi^{\rm log}(v)\bigr)	\\
&= \log(v)\expec[\kappa_T] 
- \log\bigl(\expec[\kappa_T]\bigr)\expec[\kappa_T]
+ \expec\left[\int_0^T\log\left(\frac{1}{Z_u}\right)\ud\kappa_u\right]	\\
&\quad- \log\bigl(v-\pi^{\rm log}(v)\bigr)\expec[\kappa_T] 
+ \expec\bigl[\log\bigl(\expec[\kappa_T|\G_0]\bigr)\kappa_T\bigr]
- \expec\left[\int_0^T\log\left(\frac{q^L_u}{Z_u}\right)\ud\kappa_u\right].
\end{align*}
Similarly, in view of equation \eqref{eq:opt_ut_pwr}, part (ii) of the theorem follows by solving the equation
\begin{align*}
0 &= u^{\bF}(v)-u^{\bG}\bigl(v-\pi^{\rm pwr}(v)\bigr)	\\
&= \frac{v^p}{p} \expec \left[ \int_0^T \left(Z_u \right) ^{\frac{p}{p-1}}\ud\kappa_u\right]^{1-p}	
- \frac{\bigl(v-\pi^{\rm pwr}(v)\bigr)^p}{p} \expec\left[\expec \left[ \int_0^T \left(Z_u/q^L_u \right) ^{\frac{p}{p-1}}\ud\kappa_u \bigg| \G_0 \right]^{1-p} \right].
\end{align*}
\end{proof}

\begin{proof}[Proof of Theorem \ref{thm:universal_uip}]
Note first that, since $U$ is concave, Jensen's inequality and the assumption that $S\in\Mloc(\prob,\bF)$ imply that $u^{\bF,k}(v)=U(v)$, for every utility function $U\in\mathcal{U}$ and $(k,v)\in\R^2_+$.\\
(i)$\Rightarrow$(ii):
Let $U$ be an arbitrary element of $\mathcal{U}$, $k\in\R_+$ and $v>0$. 
Consider the consumption process $c^{\bG}=(c^{\bG}_t)_{t\in[0,T]}$ given by $c^{\bG}_t=v\ind_{\{t=T\}}$, for $t\in[0,T]$. Since $\ud\kappa_u=\delta_T(\!\ud u)$ and $\prob(q^L_T=q)=1$, with $q\geq1$, Lemma \ref{lem:SuperMart=Mart} implies that $c^{\bG}\in\Cons^{\bG,k}_m((v+k)/q-k)$.
As a consequence, we have that
\[
u^{\bG,k}((v+k)/q-k)
\geq \expec\bigl[U(c^{\bG}_T)\bigr]
= U(v),
\qquad\text{ for every }v>0.
\]
On the other hand, by Jensen's inequality, for any consumption process $c\in\Cons^{\bG,k}_+((v+k)/q-k)$, it holds that
\[
\expec\bigl[U(c_T)\bigr]
\leq U\bigl(\expec[c_T]\bigr)
= U\left(q\,\expec\left[c_T/q^L_T\right]\right)
\leq U\bigl(q((v+k)/q-k+k-k/q)\bigr) = U(v),
\]
where the second inequality follows from part (i) of Lemma \ref{lem:SuperMart=Mart}, since $\qprob=\prob$ and $\ud\kappa_u=\delta_T(\!\ud u)$. We have thus shown that $u^{\bG,k}((v+k)/q-k)=U(v)=u^{\bF,k}(v)$, for every $U\in\mathcal{U}$, thus proving that (ii) holds, with the indifference value $\pi^k(v)$ being given as in \eqref{eq:universal_uip}.\\
(ii)$\Rightarrow$(iii): This implication trivially follows by  taking $k=0$ in (ii).\\
(iii)$\Rightarrow$(i):
Consider the utility functions $U_1(x)=\log(x)$ and $U_2(x)=x^p/p$, for $p\in(0,1)$. For $\mathbf{H}\in\{\bF,\bG\}$ and $i\in\{1,2\}$, denote by $u^{\mathbf{H},0}_i(v)$ the value function of the corresponding expected utility maximization problem \eqref{eq:uH}, for $v>0$ and $k=0$. 
Suppose that, for every $v>0$, there exists a value $\pi^0(v)$ such that $u^{\bG,0}_i(v-\pi^0(v))=u^{\bF,0}_i(v)=U_i(v)$, for $i\in\{1,2\}$ and all $p\in(0,1)$. In particular, this implies that $u^{\bG,0}_i(v-\pi^0(v))<+\infty$, for $i\in\{1,2\}$, and $\pi^0(v)=\pi^{{\rm log}}(v)=\pi^{{\rm pwr}}(v)$, for all $p\in(0,1)$, using the notation introduced in Proposition \ref{prop:pi}. 
The assumptions of Proposition \ref{prop:pi} are therefore satisfied and, in view of formulae \eqref{eq:ind_price_log}-\eqref{eq:ind_price_pwr}, it holds that
\[
\exp\left(\expec\left[\log(q^L_T)\right]\right)
= \expec\left[\expec\left[(q^L_T)^{\frac{p}{1-p}}\Bigr|\G_0\right]^{1-p}\right]^{1/p},
\]
for all $p\in(0,1)$. By Jensen's inequality, it holds that $\exp(\expec[\log(q^L_T)])\leq\expec[q^L_T]$.
On the other hand, the function $x\mapsto x^{1/(1-p)}$ is convex and, again by Jensen's inequality,
\[
\expec\left[\expec\left[(q^L_T)^{\frac{p}{1-p}}\Bigr|\sigma(L)\right]^{1-p}\right]^{1/p}
\geq \expec\Bigl[\expec\left[(q^L_T)^p\bigr|\G_0\right]\Bigr]^{1/p}
= \expec\left[(q^L_T)^p\right]^{1/p}.
\]
We have thus shown that
\[
\expec\left[(q^L_T)^p\right]^{1/p}
\leq \expec\left[\expec\left[(q^L_T)^{\frac{p}{1-p}}\Bigr|\G_0\right]^{1-p}\right]^{1/p}
\leq \expec[q^L_T]
\]
and $\expec[(q^L_T)^p]^{1/p}<+\infty$, for all $p\in(0,1)$. Therefore, $\expec[\expec[(q^L_T)^{\frac{p}{1-p}}|\sigma(L)]^{1-p}]^{1/p}$ converges to $\expec[q^L_T]$ as $p\rightarrow1$. In turn, this implies that
\[
v\left(1-e^{-\expec[\log(q^L_T)]}\right)
= \pi^{{\rm log}}(v) = \pi^{{\rm pwr}}(v)
= v\left(1-\expec\left[\expec\left[(q^L_T)^{\frac{p}{1-p}}\Bigr|\G_0\right]^{1-p}\right]^{-1/p}\right)
\rightarrow v\left(1-\frac{1}{\expec[q^L_T]}\right)
\]
as $p\rightarrow1$. As a consequence, it holds that $\expec[\log(q^L_T)]=\log(\expec[q^L_T])$. Since the function $x\mapsto\log(x)$ is strictly concave, this implies that there exists a strictly positive constant $q$ such that $\prob(q^L_T=q)=1$. The fact that $q\geq1$ follows since $\expec[1/q^L_T]\leq1$, by the supermartingale property of $1/q^L$ on $\basisgp$.

It remains to show that, under conditions (i), (ii), or (iii), the optimal wealth process in problem \eqref{eq:uH} for $\mathbf{H}=\bG$, denoted by $V^{\bG}=(V^{\bG}_t)_{t\in[0,T]}$, is given as in \eqref{eq:universal_wealth}. The optimal consumption plan $c^{\bG}$ constructed in the first part of the proof belongs to $\Cons_m^{\bG,k}((v+k)/q-k)$. Therefore,
\[
\frac{v}{q} 
= \expec\left[\frac{c^{\bG}_T}{q^L_T}\bigg|\G_t\right]
= \frac{V^{\bG}_t+k}{q^L_t}-\expec\left[\frac{k}{q^L_T}\bigg|\G_t\right]
= \frac{V^{\bG}_t+k}{q^L_t}-\frac{k}{q}
\qquad\text{a.s. for all }t\in[0,T],
\]
where the second equality follows from the definition of the set $\Cons_m^{\bG,k}((v+k)/q-k)$.
\end{proof}

\begin{proof}[Proof of Proposition \ref{prop:uip_bounds}]
Similarly as in the proof of Theorem \ref{thm:universal_uip}, it holds that $u^{\bF,k}(v)=U(v)$, for every $U\in\mathcal{U}$, $k\in\R_+$ and $v\geq0$.
The consumption process $c^{\bG}=(c^{\bG}_t)_{t\in[0,T]}$ defined by $c^{\bG}_t=v\ind_{\{t=T\}}$, for $t\in[0,T]$, belongs to $\Cons_+^{\bG,k}((v+k)/q_{\min}-k)$. Indeed, under the present assumptions it holds that $\expec[(v+k)/q^L_T|\G_0]\leq (v+k)/q_{\min}$ a.s. Therefore, for all $k\in\R_+$ and $v>0$, we have that
\[
u^{\bF,k}(v) = U(v) = \expec[U(c^{\bG}_T)] \leq u^{\bG,k}((v+k)/q_{\min}-k),
\]
which implies that $v-\pi^{U,k}(v)\leq (v+k)/q_{\min}-k$, thus proving the first inequality in \eqref{eq:uip_bounds}.
Consider then an arbitrary consumption process $c=(c_t)_{t\in[0,T]}\in\Cons_+^{\bG,k}((v+k)/q_{\max}-k)$. By Jensen's inequality, it holds that
\[
\expec[U(c_T)] \leq U(\expec[c_T]) \leq U\left(q_{\max}\,\expec\left[\frac{c_T}{q^L_T}\right]\right) 
\leq U(v) = u^{\bF,k}(v),
\]
where the third inequality follows from Lemma \ref{lem:SuperMart=Mart}. By the arbitrariness of $c$, this implies that $u^{\bG,k}((v+k)/q_{\max}-k)\leq u^{\bF,k}(v)$, thus showing that $v-\pi^{U,k}(v)\geq (v+k)/q_{\max}-k$.
\end{proof}

\section{Conclusions}	\label{sec:conclusions}

In this paper, we have presented a general study of the value of informational arbitrage, in the context of a semimartingale model of a complete financial market with inside information.
In our analysis, the assumption of market completeness plays a central role. In particular, it can be naturally transferred to the initially enlarged filtration $\bG$, thus enabling us to obtain a precise characterization of the validity of NUPBR and NFLVR in the insider financial market $(\Omega,\bG,\prob;S)$. In turn, this provides a general and simple duality approach to the solution of optimal consumption-investment problems. In the case of typical utility functions, market completeness leads to fully explicit solutions, which reveal interesting features of the value of informational arbitrage.

The value of informational arbitrage can also be defined and studied in general incomplete markets. In particular, the existence and uniqueness result of Theorem \ref{thm:IndPrice} still holds in incomplete markets, as long as the optimal investment-consumption problem in $\bG$ is well-posed. More precisely, if the primal and dual value functions in $\bG$ are finite and $(\Omega,\bG,\prob;S)$ satisfies NUPBR (but not necessarily NFLVR), then the results of \cite{CCFM2015} imply that the value function is sufficiently regular to prove the existence and uniqueness of the value of informational arbitrage.
However, except for specific models, one cannot obtain an explicit description of the value of informational arbitrage. 
Furthermore, in general incomplete markets, there does not exist a simple criterion for determining whether the inside information generates arbitrage in $\bG$ (compare with Theorem \ref{thm:results_NUPBR}).

In the present work, we have considered a frictionless financial market where the traded assets are infinitely liquid. In our view, an interesting direction of future research consists in studying the value of informational arbitrage under more realistic market structures featuring transaction costs as well as price impact, in the spirit of \cite{KHS06}.
In those cases, it may be optimal for an informed agent to underinvest in the arbitrage opportunity, because the trading activity itself can affect the profitability of arbitrage strategies or simply because of the presence of market frictions.


\bibliographystyle{alpha}
\bibliography{value_arbitrage_biblio}
\end{document}